\documentclass{article}
\usepackage{csquotes}
\usepackage{float}
\usepackage{setspace}
\usepackage{wrapfig}
\usepackage{mathtools}
\usepackage{geometry}
\usepackage{natbib}
\usepackage{pifont}
\usepackage{enumitem}
\usepackage{amsmath}
\usepackage{dsfont}
\usepackage{amssymb}
\usepackage{amsthm}
\usepackage{nth}
\usepackage{booktabs}
\usepackage{titletoc}
\usepackage{thm-restate,thm-autoref}
\usepackage{scalerel,stackengine}
\makeatletter
\def\moverlay{\mathpalette\mov@rlay}
\def\mov@rlay#1#2{\leavevmode\vtop{%
   \baselineskip\z@skip \lineskiplimit-\maxdimen
   \ialign{\hfil$\m@th#1##$\hfil\cr#2\crcr}}}
\newcommand{\charfusion}[3][\mathord]{
    #1{\ifx#1\mathop\vphantom{#2}\fi
        \mathpalette\mov@rlay{#2\cr#3}
      }
    \ifx#1\mathop\expandafter\displaylimits\fi}
\makeatother

\newcommand{\bigcupdot}{\charfusion[\mathop]{\bigcup}{\cdot}}
\usepackage[linesnumbered,ruled,vlined]{algorithm2e}
\usepackage[english]{babel}
\usepackage{tikz}
\RequirePackage[colorlinks,citecolor=black,linkcolor=blue,urlcolor=blue,pagebackref]{hyperref}
\usepackage{cleveref}
\usetikzlibrary{positioning}
\tikzset{main node/.style={circle,fill=blue!20,draw,minimum size=1cm,inner sep=0pt}, }
\tikzset{
main node/.style=
    {circle,fill=blue!20,draw,minimum size=1cm,inner sep=0pt},
root node/.style=
    {circle,fill=red!20,draw,minimum size=1cm,inner sep=0pt},
outcome node/.style={
    fill=gray!20,
    draw,
    minimum size=1.5cm,
    inner sep=0pt},
graph directed edge/.style={
    ->,
    >=stealth,
    thick,
  }, 
  graph twosided edge/.style={
    <->,
    >=stealth,
    ultra thick,
  },
  graph forward edge/.style={
    graph directed edge,
    every edge/.style={
      edge node={node [fill=white] {$\phi$}},
      densely dotted,
      draw,
    },
  },
graph implies edge/.style={
    graph twosided edge,
    every edge/.style={
      edge node={node [fill=white] {CP}},
      draw, ultra thick,
    },
  },}

\theoremstyle{plain}
\newtheorem{proposition}{Proposition}
\newtheorem{lemma}{Lemma}
\newtheorem{theorem}{Theorem}
\newtheorem{corollary}{Corollary}
\theoremstyle{remark}
\newtheorem*{remark}{Remark}

\newtheorem*{definition}{Definition}
\newtheorem*{example}{Example}
\newcommand{\phik}{\phi_{V\text{-}k}}
\newcommand{\phispa}{\phi_{V\text{-}1}}

\newcommand{\R}{\mathbb{R}}
\newcommand{\score}{\operatorname{score}}

\newcommand{\FPA}{\operatorname{FPA}}

\newcommand{\W}{\operatorname{W}}
\newcommand{\LQL}{\operatorname{LQL}}
\renewcommand{\L}{\operatorname{L}}
\newcommand{\LQW}{\operatorname{LQW}}

\newcommand{\median}{\operatorname{median}}
\DeclareMathOperator*{\argmax}{arg\,max}

\newcommand{\ASCJ}{\operatorname{asc-j}}
\newcommand{\ODESCJ}{\operatorname{odesc-j}}
\newcommand{\children}{\operatorname{children}}

\newcommand{\thresh}{\operatorname{thresh}}

\newcommand{\tsucc}{\operatorname{succ}}
\newcommand{\pred}{\operatorname{pred}}
\newcommand{\btheta}{\boldsymbol\theta}
\newcommand{\bsigma}{\boldsymbol{\sigma}}

\newcommand{\bTheta}{\boldsymbol\Theta}
\newcommand{\1}{\mathds{1}}

\newcommand{\RI}{\operatorname{RI}}
\newcommand{\reveal}{\operatorname{reveal}}

\newcommand{\DESC}{\operatorname{desc}}
\newcommand{\CP}{\operatorname{CP}}

\begin{document}

\title{Contextually Private Mechanisms}
\author{Andreas Haupt and Zo\"{e} Hitzig\footnote{Haupt: Stanford Digital Economy Lab, h4upt@stanford.edu. Hitzig: Harvard Society of Fellows, zhitzig@g.harvard.edu. We thank especially Ben Golub, Eric Maskin, and Shengwu Li. We also thank Yannai Gonczarowski, Jerry Green, Chang Liu, Dirk Bergemann, Stephen Morris, Mohammad Akbarpour, Ilya Segal, Paul Milgrom, Kevin He, Navin Kartik, Ludvig Sinander, Andrew Mackenzie, Clayton Thomas, Oghuzan Celebi, David Laibson, Modibo Camara, Katrina Ligett, Benji Niswonger and Max Simchowitz, as well as audiences in seminars, workshops and conferences (Harvard and MIT Theory Lunch workshops; seminars at Oxford, Stanford GSB, Boston College, NYU and Boston University; Penn State Economic Theory Conference, Penn/Wharton Mini Conference in Economic Theory; Stony Brook Game Theory Festival, Berkeley SLMath Workshop on Mechanism and Market Design). An extended abstract of an earlier version of this paper---comprised largely of the results now in \Cref{sec:full_cpimpossibilities} and \Cref{sec:generic}---appeared in Proceedings of the 23rd ACM Conference on Economics and Computation (EC 2022). This work was supported by a Microsoft Research PhD Fellowship (Hitzig).}}  

\maketitle

\singlespacing

\begin{abstract}
We introduce a framework for comparing the privacy of different mechanisms. A mechanism designer employs a dynamic protocol to elicit agents' private information. Protocols produce a set of \emph{contextual privacy violations}---information learned about agents that may be superfluous given the context. A protocol is \emph{maximally contextually private} if there is no protocol that produces a proper subset of the violations it produces, while still implementing the choice rule. Contextual privacy violations arise when a choice rule makes some agents collectively, but not individually, pivotal. In auctions, designing for contextual privacy requires choosing an initial question posed to each agent and the order in which agents are queried. We study a particular maximally contextually private protocol for $k$-item Vickrey auctions---the \textit{ascending-join protocol}---and show that it achieves maximal contextual privacy by delaying queries to bidders whose privacy it protects.
\end{abstract}

\section{Introduction}\label{sec:intro}

\onehalfspacing

In standard mechanism design, a designer elicits reports of agents' private information to determine the outcome of a social choice rule. Ex-post, in a direct revelation mechanism, the designer learns \emph{all} of agents' private information---typically more than is necessary for computing the rule. For example, in a sealed-bid second-price auction, the designer observes every submitted bid, even though to determine the outcome it suffices to know only who has the highest bid and the value of the second-highest bid. The exact values of losing bids are superfluous, and even the winner’s bid need not be known precisely.

Moreover, there are several reasons why it may be preferable for the designer not to learn \enquote{too much} about agents. First, gaining excess knowledge of agents' private information could expose the designer to legal liabilities or political risk. For example, as recounted in \citet{mcmillan1994selling}, a second-price auction for spectrum licenses in New Zealand had a \enquote{political defect}: \enquote{by revealing the high bidder's willingness to pay, the auction exposed the government to criticism} \citep{mcmillan1994selling}. The government would have benefited from a design that ensured it only learned what was strictly needed---with the idea that, as the adage goes, what they don't know can't hurt them. Second, if agents have privacy concerns, they may be reluctant to reveal their private information even if the appropriate allocative incentives are in place. Agents may worry that their private information could be used against them in subsequent interactions with the designer or third parties \citep{rothkopf1990vickrey, cramton1998ascending, ausubel2004efficient}.

In this paper, we study how mechanism designers can limit the superfluous information they learn. In our setup, when a designer commits to a social choice rule, they also choose a dynamic protocol for eliciting agents' information (or \enquote{types}). These dynamic protocols allow the designer to learn agents' private information gradually, ruling out possible type profiles until they know enough to compute the outcome of the rule. The key idea we introduce is a \emph{contextual privacy violation}. A protocol produces a contextual privacy violation for a particular agent at a particular type profile if the designer learns a piece of the agent's private information that is superfluous, i.e. it is not necessary for computing the outcome. We study protocols that are on the frontier of contextual privacy and implementation---that is, we study those protocols that produce an inclusion-minimal set of privacy violations subject to computing a choice rule. For some choice rules, it is possible to find protocols that produce no violations for any agent at any type profile---we call these protocols \emph{fully contextually private}.

Our set up describes extensive-form games in which the designer cannot use a trusted mediator or any mediating technology like anonymization or other cryptographic tools.\footnote{Advanced cryptography is often costly in terms of time, money or computational power. Even if available, some sophisticated solutions may be wasteful---in one of the earliest large-scale uses of secure multi-party computation, a double auction with sugar beet farmers in Denmark, designers wondered \enquote{if the full power of multiparty computation was actually needed,} or if a simpler implementation guided by a weaker privacy criterion may have sufficed \citep{bogetoft2009secure}.}  We start here not only for theoretical reasons---it is a minimal assumption to make about the environment\footnote{In our setup, the designer learns information if and only if the agent discloses it. Mediating technologies represent different commitments to forgetting parts of what the agent says. Anonymization allows the designer to commit to forgetting sender identifiers. From the participants' perspective, a trusted third party relaying only an auction's winner and price represents a commitment that all irrelevant information will be forgotten. Studying minimal assumptions---no commitment---can explain the historical prevalence of dynamic auction protocols (compare chapter 7 in \citet{Learmount1985history}) before cryptographic technologies were available, and can inform the design of auctions in futures where they might not be available anymore, e.g., due to quantum computing.}---but also because many real-world auctions operate dynamically, often with no mediating technology. Live ascending formats remain standard in the markets for fine art, wine, and antiques ({e.g.}, at Sotheby's); descending-price protocols, also live, are used in the cut-flower auctions ({e.g.}, at Aalsmeer market) and several European fish markets \citep{fluvia2012buyer}; the U.S. Forest Service sells timber in oral ascending auctions; and multi-round clock auctions have become essential for radio-spectrum allocation, natural-gas capacity sales, and wholesale-electricity procurements. In all of these settings, there are reasons why the auctioneer might want to limit what they learn about bidders types---post-auction price discrimination, distortions of subsequent competition due to revealed strategic information, and reduced participation by bidders unwilling to disclose proprietary data.\footnote{\citet{cramton1998ascending} and \citet{ausubel2004efficient}, for example, discuss such considerations in the context of clock auctions specifically. \citet{rothkopf1990vickrey} discusses the value of winner privacy in settings with post-auction negotiations, such as auctions for utilities, oil- and coal-leases, and construction contracts.} In the case of the 2017 Federal Communications Incentive Auction, for example, privacy guarantees were seen as valuable since they ``may alleviate winners’ concerns about misuse of the revealed information\dots and allow the auctioneer to conceal winners’ politically sensitive windfalls" \citep{milgrom2020clock}.

The rest of the paper is structured as follows. We discuss related literature in \Cref{sec:lit}. In \Cref{sec:model}, we articulate our formal framework, present the key definitions of the paper, and compare our approach to alternatives. We first define the dynamic messaging game played by agents. Then, we introduce the central definitions: contextual privacy violations, the contextual privacy order (an inclusion order based on violations), and maximally contextually private protocols (maximal elements in the inclusion order). We then further situate contextual privacy, articulating the real-world concerns it may capture, and discussing how it complements existing analyses of information disclosure in mechanism design. Finally we make a methodological observation that simplifies our analysis---\Cref{prop:partitional} shows that our analyses can proceed entirely in terms of the partitions of the type space. 

In \Cref{sec:generic}, we study how specific properties of choice rules lead to contextual privacy violations. In our first result, \Cref{thm:characterization}, we provide a characterization of choice rules that fully avoid contextual privacy violations. The characterization shows a connection between contextual privacy and pivotality---if some group of agents are jointly pivotal but no agent is individually pivotal on a subset of the type space, there must be a contextual privacy violation, and the designer must decide whose privacy to protect. Generally, this characterization shows that few choice rules avoid contextual privacy violations entirely. We use the characterization to show that the first-price auction rule (\Cref{prop:fbacp}) and the serial dictatorship choice rule (\Cref{prop:sd_cp}) admit protocols that produce no contextual privacy violations. Two necessary conditions for contextual privacy, \Cref{cor:collectindividual} and \Cref{cor:corners}, are direct corollaries of the characterization. We use these necessary conditions to rule out the existence of fully contextually private auctions for $k$-item Vickrey auctions (\Cref{prop:spanoncp}) and any stable choice rule for the school choice problem (\Cref{thm:stabsecp}).\footnote{We use the necessary conditions in \Cref{sec:full_cpimpossibilities} to show further impossibilities---we rule out the existence of fully contextually private efficient allocations in housing assignment and generalized median voting rules.}
 
In \Cref{sec:maxcp}, we take the perspective of a privacy-conscious designer who wants to implement a choice rule through a maximally contextually private protocol. First we show that for a class of social choice rules on totally ordered type spaces that contains $k$-item Vickrey auctions, it is without loss only to consider protocols which consist of threshold queries that are monotonically increasing or decreasing in the threshold after an initial guess (\Cref{thm:bimon}). This result delivers both theoretical and practical insight. On the theoretical side, it serves as a valuable reduction that allows us to prove that certain protocols are maximally contextually private. On the practical side, it highlights that there are two key dimensions of designing for contextual privacy: choosing what initial query to ask to each agent, and choosing the order in which to query agents. 

We then turn to the special case of ascending protocols for $k$-item Vickrey auction rules as an instructive and practically relevant example---we see how maximally contextually private protocols for the second-price auction rule choose a set of agents to protect, and delay asking questions to the protected agents. In \Cref{thm:maxcpkpa}, we show that an ascending-join protocol is maximally private. We show that it can be viewed as repeatedly asking agents whether they can rule out a particular outcome and formalize the sense in which it protects agents' privacy by delaying queries to them (\Cref{prop:delay}).\footnote{In \Cref{sec:overdesc}, we show an analogous result for the (over-)descending protocol for the $k$-item Vickrey auction \cite{overdescending}.}

We study variations of contextual privacy in \Cref{sec:extensions}. \textit{Group contextual privacy} requires that when the designer learns something about the entire type profile, it must make a difference to the outcome. We show that a protocol is group contextually private if every query ``rules out" an outcome (\Cref{prop:groupcp}). \textit{Individual contextual privacy} requires that if two type profiles that differ for one agent are distinguished, they need to lead to different allocations for this agent, presuming a private allocation domain. We show that the set of individual contextual privacy violations is the union of contextual privacy violations and nonbossiness violations (\Cref{lem:icp_nonbossy}). We conclude in \Cref{sec:conclusion}, highlighting several further avenues for future research. 

Proofs omitted from the main text can be found in \Cref{sec:proofs}. Additional impossibility results for full contextual privacy are in \Cref{sec:full_cpimpossibilities}. A maximally contextually private descending protocol for the $k$-item Vickrey auction is discussed in \Cref{sec:overdesc}. In \Cref{sec:informativeness}, We compare contextual privacy and relative informativeness.

\section{Related Literature}\label{sec:lit}

Two important antecedents to our notion of contextual privacy appear in the literature on decentralized computation: \emph{unconditional privacy} and \emph{perfect implementation}. Unconditional full privacy, introduced by \citet{chor1989zero} and further explored by \citet{brandt2005unconditional, brandt2008existence}, requires that no information beyond the outcome itself is revealed by a protocol's execution, independent of cryptographic or mediator assumptions. Within centralized mechanism design, \citet{milgrom2020clock}'s \textit{unconditional winner privacy} narrows this requirement, protecting solely the winner's valuation in an auction setting. While these concepts highlight critical privacy objectives, they are binary in nature---either fully met or not. Our framework introduces a richer comparative language: contextual privacy defines violations at the granular level of individual agents and type profiles, facilitating systematic privacy comparisons across mechanisms. Thus, our contribution generalizes unconditional winner privacy by allowing designers to articulate nuanced privacy objectives beyond protecting only the winner's information. 

The notion of \emph{perfect implementation}
\citep{izmalkov2005rational,izmalkov2011perfect} seeks implementations that do not rely on trusted mediators, but rather rely on technologies that enable verification of what was learned---like sealed envelopes and nested envelopes. By contrast, we consider privacy protection under more minimal assumptions about the technology available to the designer, and offer a privacy order rather than an implementation concept. Perfect implementation is one of many cryptographic approaches to auctions---for a survey of the literature on cryptographic auctions, see \citet{alvarez2020comprehensive}.

Several papers consider mechanism design through the lens of data minimization, studying how to reduce the ``volume" or ``amount" of information exchanged to achieve the designer's objectives. \citet{mackenzie2022menu} introduce dynamic \emph{menu mechanisms} and demonstrate, using their relative informativeness criterion, that such mechanisms typically disclose less information than direct-revelation mechanisms.\footnote{The connection between contextual privacy and relative informativeness is further discussed in \Cref{sec:informativeness}.} Relatedly, \citet{segal2007communication,segal2010communication} and \citet{gonczarowski2019} focus on the minimal communication complexity necessary for implementing specific social-choice rules. Other papers go beyond consideration of the ``volume" of information communicated, and into quantitative measures of \enquote{privacy loss} \citep{eilat2021bayesian, liu2020preserving}, where privacy loss is defined as some measure ({e.g.}, Shannon entropy, Kullback-Leibler divergence) of information revelation (comparing the designer's prior to their posterior).  Our approach differs from these by explicitly evaluating privacy at the level of individual agents and economic context: we classify disclosures according to their necessity for the implementation of the given social-choice rule, rather than solely by the total quantity of information revealed. In other words, these measure-based criteria treat all datum as equal, contextual privacy, treats information differently depending on how it is used. Relatedly,  \citet{StrackYang2024_Ecta_PrivacySignals} and \citet{HeSandomirskiyTamuz} study  information design under privacy constraints that also treat different pieces of information differently---requiring posterior beliefs about some “protected” coordinates (privacy sets in the former paper, other agents’ signals in the latter paper) to be unchanged after signal realization---but their constraints are not derived from a choice rule. In this respect, contextual privacy also differs from differential privacy \citep{dwork2006differential}. Where contextual privacy justifies information revelation through its relevance to a social outcome, differential privacy explicitly restricts the sensitivity of the social outcome to revealed information. For a survey of differential privacy's incorporation into mechanism design, see \citet{pai2013privacy}. 

Finally, our analysis relates to a literature in mechanism design concerned with settings where the designer may exploit the information elicited from participants. A few papers study auditability of mechanisms, asking how much of participants' information is needed to verify the outcome of a mechanism ex-post \citep{grigoryan2023theory,woodward2020self, pycia2024ordinal}. The notion of credibility \citep{akbarpour2020credible} requires that the designer not have any profitable deviation from the announced procedure, so that bidders can safely reveal whatever information is requested. Contextual privacy is motivated by the similar concerns---participants' reluctance to reveal sensitive data when they doubt the designer---but addresses a distinct question: conditional on faithful implementation, how much information beyond that needed for the rule is unnecessarily exposed?\footnote{Sometimes the desiderata happen to coincide, for example, a descending auction is simultaneously credible and fully (and thus maximally) contextually private. } 

\section{Model}\label{sec:model}
We now define the model. The environment consists of:
\begin{enumerate}[itemsep=0pt, parsep=0pt]
    \item A finite set of players, $N = \{0, 1, 2, \dots, n\}$. We call player $0$ the \emph{designer}, and players $1, 2, \dots, n$ \emph{agents}.
    \item A finite set of outcomes, $X$.
    \item A finite space of agent type profiles, $\bTheta = \bigtimes_{i=1}^n \Theta_i$.\footnote{We sometimes assume symmetric type spaces for ease of exposition, in which case we denote the individual type spaces $\Theta=\Theta_i$ for all $i$.} 
    \item Agent utilities $u_i \colon \Theta_i \times X \to \R_+$.
    \item A social choice function $\phi\colon \bTheta \to X$ that the designer would like to compute.
\end{enumerate}

In order to compute $\phi$, the designer elicits information with a \emph{mechanism}, which we model as an (perfect recall, finite-depth) extensive-form game. Formally, a mechanism is a 6-tuple $P = (H, i, (A_i)_{i=1, 2, \dots, n}, (\mathcal I_i)_{i=0, 1, \dots, n}, A, I)$. See \Cref{tab:notation} for an explanation of the different components of $P$. The mechanism terminates whenever the designer chooses an action and every path through the game tree ends with a designer action. The actions available to the designer are $x \in X$. We call the actions taken by the agents \emph{messages}, the action taken by the designer an \emph{allocation}. We assume individual private values, i.e. information sets of agents $i = 1, 2, \dots, n$ only depend on their type $\theta_i$ and the set of messages sent prior to history $v$. The information sets of the designer are given by all prior messages sent.

\begin{table}[!ht]\footnotesize
    \centering
        \caption{Notation for Dynamic Protocols.}
        \label{tab:notation}
    \begin{tabular}{ccc} 
    \toprule
    Name & Notation & Representative Element \\
    \midrule
    histories & $H$ & $h$ \\
    player called to play at $h$ & $i(h)$ & \\
    actions & $A_i$ & $a_i$ \\
    information sets of player $i$ & $\mathcal I_i$ & $I_i$ \\
    messages available at $I_i$ & $A(I_i)$ &\\
    information set at history $h$ & $I(h)$ & \\
    strategy of player $i$ & $\sigma_i(I_i)$ \\
    protocol & $(P, \bsigma)$ \\
    \bottomrule
    \end{tabular}
\end{table}

Players choose actions according to \emph{strategies}, which are deterministic mappings $\sigma_i \colon \mathcal I_i \to A_i$ such that $\sigma_i(I_i) \in A(I_i)$. We denote a \emph{strategy profile} by $\bsigma = (\sigma_i)_{i \in N}$. Call $P$ a \emph{mechanism}, and the tuple $(P, \bsigma)$ a \emph{protocol}. We assume that for all type profiles $\btheta$, the allocation chosen by the designer at the end of the protocol is $\phi(\btheta)$.

\subsection{Contextual Privacy of Protocols}\label{subsec:cpdef}

We say the designer \emph{distinguishes} two type profiles $\btheta, \btheta' \in \bTheta$ with protocol $(P, \bsigma)$ if the two type profiles reach different information sets for the designer at the end of the game.

The central definition of the paper is that of a \emph{contextual privacy violation}.
\begin{definition}
    Let $(P, \bsigma)$ be a protocol that computes a social choice function $\phi$. We say that agent $i$ has a \emph{contextual privacy violation} with respect to $(P, \bsigma)$ at $\btheta = (\theta_i, \btheta_{-i}) \in \bTheta$ if there is $\theta_{i}' \in \Theta_i$ such that $\btheta$ and $(\theta_i', \btheta_{-i})$ are distinguished but $\phi (\btheta) = \phi(\theta_i', \btheta_{-i})$. Otherwise, we say that $i$'s privacy is \emph{preserved} at $\btheta$. We denote the set of type profile-agent pairs at which contextual privacy violations of a protocol $(P, \bsigma)$ occur by $\Gamma( P, \bsigma ) \subseteq \bTheta \times N$.
\end{definition}
In other words, there is a contextual privacy violation for agent $i$ under protocol $P$ that computes choice rule $\phi$ at type profile $\btheta$ if the designer can tell apart types $\theta_i, \theta_i'$ but $\theta_i$ and $\theta_i'$ lead to the same outcome, holding fixed $\btheta_{-i}$. 

If $\Gamma(P, \bsigma) = \emptyset$, we call $(P, \bsigma) $ \emph{fully contextually private for $\phi$}---and if a choice rule $\phi$ can be implemented with a fully contextually private protocol, we say the choice rule $\phi$ is fully contextually private. If $\Gamma(P, \bsigma)$ is inclusion-minimal among all protocols $(P, \bsigma)$ that compute social choice rule $\phi$, we say that $(P, \bsigma)$ is a \emph{maximally contextually private} protocol for $\phi$. Two protocols $(P, \bsigma)$ and $(P', \bsigma')$  that both compute $\phi$ are  \emph{contextual privacy-equivalent} for $\phi$ if $\Gamma(P, \bsigma) = \Gamma(P', \bsigma')$. 

Asking whether a protocol has good contextual privacy properties and good incentive guarantees are mostly orthogonal questions. Once protocols with good contextual privacy properties are identified, their incentive properties can be analyzed on a case-by-case basis. We do this for the ascending-join auction, which we study in \Cref{sec:maxcp}. In particular, we will use the demanding notion of \emph{obvious dominance}. For a protocol $P$ and a designer strategy $\sigma_0$, we denote for agent $i$, a history $h \in H$, and an action $a \in A(h)$,  $X_{h, a, \sigma_0, \sigma_i}$ to be the set of all outcomes $x\in X$ reachable from history $h$ by any actions from agents $i' \in N \setminus \{0, i\}$ and following $\sigma_0$ and $\sigma_i$ at histories $h \in H$ where $i(h) \in \{0, i\}$. 
\begin{definition}
   For an agent $i$, a strategy $\sigma_i \colon H \to A$ is \emph{obviously dominant} for protocol $P$ if at all histories $h \in H$ at which agent $i$ is asked to play, $i(h) = i$, and for all actions $a' \in A(h) \setminus \{\sigma_i(I(h))\}$,
    \[
    \inf_{x \in X_{h, a, \sigma_0, \sigma_i}} u_{i} (x;\theta_i) \ge \sup_{x' \in X_{h, a', \sigma_0, \sigma_i}} u_{i} (x'; \theta_i),
    \]
We say that a protocol $(P, \bsigma)$ for $\phi$ \emph{implements} $\phi$ in \emph{obviously dominant strategies} if $\sigma_i$ is obviously dominant for all agents $i = 1, 2, \dots, n$.
\end{definition}

\subsection{Contextual Privacy of Iterative Partitions}

For our analysis of privacy properties, only the information about the type profile that the designer can infer from the realized play is relevant. Thus two protocols are indistinguishable for our purposes whenever they implement the same social choice function and induce the same information partition for the designer. This observation permits a substantial reduction of the set of protocols that must be considered. Instead of arbitrary extensive‑form mechanisms we can consider protocols induced by \emph{iterative partitions}, which sequentially refine the knowledge of the designer. Iterative partitions induce what we call \emph{partitional protocols}. Note that this reduction is somewhat analogous to the revelation principle: any protocol $(P, \bsigma)$ can be simulated by a partitional protocol that elicits exactly the information required to reach the same final partition. 

We first define an \emph{iterative partition} of the space of type profiles $\bTheta$. This iterative partition, $P = (V, E)$, is a directed rooted tree whose nodes are subsets of the type space, i.e. $V \subseteq 2^{\bTheta}$, and whose root node is the type space $\bTheta$. We will denote a representative node of $P$ by $\tilde \bTheta \subseteq \bTheta$. For every node $\tilde \bTheta$ in the iterative partition, there is an agent $i(\tilde \bTheta)$ such that the children of $\tilde \bTheta$ form a partition of $\tilde \bTheta$ based on type $\theta_i$.\footnote{Formally, for any children $\tilde \bTheta', \tilde \bTheta''$ of $\tilde \bTheta$, $\tilde \bTheta' \cap \tilde \bTheta'' = \emptyset$,  $\bigcup_{\tilde \bTheta' \in \children(\tilde \bTheta)} \tilde \bTheta'= \tilde \bTheta$ (the children form a partition of the parent), and $\tilde \bTheta' = \tilde \bTheta \cap (\tilde\Theta_{i}' \times \bTheta_{-i})$ for some $\tilde \Theta_{i}' \subseteq \Theta_i$ (when comparing the parent and the child, only agent $i$'s possible types change).} We say that an iterative partition \emph{computes} $\phi$ if for every leaf $\tilde \bTheta$ of $P$, there is an outcome $x \in X$ such that for all $\btheta \in \tilde \bTheta$,  $\phi(\btheta) = x$. We call the partition induced by leaves of the tree $P$ the \emph{final partition}. $P$ distinguishes $\btheta, \btheta' \in \bTheta$ if $\btheta, \btheta'$ lie in different cells of the final partition.

An iterative partition $P$ and a social choice function $\phi$ together define a \emph{partitional protocol} $((H, i, (A_i)_{i=1, 2, \dots, n},A, (\mathcal I_i)_{i=0, 1, \dots, n}, I), \bsigma)$ for $\phi$ as follows:
\begin{itemize}[itemsep=0pt]
    \item Histories contain agent types and partition sets, i.e. $H = \{ (\btheta, \tilde \bTheta) : \btheta \in \bTheta, \tilde \bTheta \subseteq \bTheta \}$.
    \item The agent called to play at history $h = (\btheta, \tilde \bTheta)$ is the agent whose type the partition changes along from $\tilde \bTheta$ to $\children (\tilde \bTheta)$.
    \item Actions for agent $i$ are subsets of their type space, $A_i = 2^{\Theta_i}$. The available actions to the agent called to play are the $\tilde \Theta'_i$ for $\tilde \bTheta' \in \children(\tilde \bTheta)$.
    \item The designer's information set at history $h = (\btheta, \tilde \bTheta)$ is $I(h) = \tilde \bTheta$. Any information set of agent $i$ at history $h = (\btheta, \tilde \bTheta)$ contains at least $\theta_i$, and in \Cref{prop:fbacp} and \Cref{sec:maxcp}, we will assume it is $(\theta_i, \tilde \bTheta)$.\footnote{The agent information beyond knowing their own private information is inconsequential for many of our results, except those pertaining to incentives.}
    \item The strategies of the induced protocol are as follows: The designer's strategy is $\sigma_0 (\tilde \bTheta) = x$ for $\{x \} = \phi(\tilde \bTheta)$ for any leaf $\tilde \bTheta$. When agent $i = 1, 2, \dots, n$ is called to play at $ h = (\btheta, \tilde \bTheta)$, their strategy is $\sigma_i(I(h)) = \tilde \Theta_{i}'$ with $\theta_i \in \tilde \Theta_{i}'$, $\tilde \bTheta' \in \children(\tilde \bTheta)$.
\end{itemize}
The rooted tree $(V,E)$ induces a natural precedence order on nodes: for $v,v'\in V$ we write $v \prec v'$ if $v$ is a strict ancestor of $v'$ along the unique path from the root. Whenever we refer to the ``earliest'' node or to a node that is ``minimal in precedence order,'' we mean minimal with respect to this partial order.

We say that protocols $(P, \bsigma)$ and $(P', \bsigma')$ are \emph{equivalent} if they compute the same social choice function and $(P, \bsigma)$ distinguishes $\btheta$ from $\btheta'$ if and only if $(P', \bsigma')$ distinguishes $\btheta$ from $\btheta'$. 

\begin{proposition}
\label{prop:partitional}
    For every protocol $(P, \bsigma)$, there exists an iterative partition whose induced partitional protocol $(P', \bsigma')$ is equivalent to $(P, \bsigma)$.
\end{proposition}
\begin{proof}
Let $(P, \bsigma)$ be a protocol, with $P = (H, i, (A_i)_{i=1, 2, \dots, n},A, (\mathcal I_i)_{i=0, 1, \dots, n}, I)$. 

We construct a partition whose partitional protocol is equivalent to $(P, \bsigma)$. It is without loss to consider a minimal set of histories $\tilde H$, which can be reached from some type profile $\btheta \in \bTheta$, under play according to $\bsigma$, for $P$. As $\bsigma$ consists of deterministic strategies, $\tilde H$ can be encoded by $(\btheta, \bTheta_{\tilde h})_{\tilde h \in \tilde H}$, where $\bTheta_{\tilde h} \subseteq \bTheta$ is the set of type profiles consistent with $\bsigma$ and the play at histories $\tilde h' \prec \tilde h$ preceding $\tilde h$. As agent information sets only depend on their own type and the history of play, $\bTheta_{\tilde h}$ refines only a single agent's type along histories $\tilde h$. Hence, the rooted tree of histories $(\bTheta_{\tilde h})_{\tilde h \in \tilde H}$ where edges are between histories directly preceding each other, forms an iterative partition. As $(P, \bsigma)$ is a protocol for $\phi$, the rooted tree of histories is also an iterative partition \emph{for $\phi$}. To see this, consider a leaf $\tilde \bTheta$ of this iterative partition. This leaf corresponds to an information set $I(\tilde \bTheta)$ of the designer, and results in outcome $\sigma_0(I(\tilde \bTheta))$. As $(P, \bsigma)$ is a protocol for $\phi$, it must be that all type profiles $\btheta \in \bTheta$ that lead to information set $I(z)$ have $\phi (\btheta) =\sigma_0(I(z))$, showing that the iterative partition $P$ induces a partitional protocol \emph{for $\phi$}. Given this representation of histories, the agents' actions can be relabeled as well. Consider a history $h = (\btheta, \bTheta_h)$ where agent $i(h)$ is called to play. Recall that agent information sets only depend on the agent's type in addition to the history of messages sent, so we can write $I(\btheta, \bTheta_h) = I(\theta_i, \bTheta_h)$. Hence, the agent can, instead of sending message $a = \sigma_{i(h)} (I(\btheta, \bTheta_h))$ from $A(h)$ send the subset of types $\theta_i$ which would have led to $a$. Thus,  agent actions can be encoded as subsets of the agent's type space, $2^{\Theta_i}$, leading to an equivalent protocol.
\end{proof}

The following observation allows us, for the purposes of privacy, to consider only iterative partitions. It relates the contextual privacy violations of partitional protocols to those of the underlying iterative partition.

\begin{lemma}
\label{lem:partitionalcpviolations}
    Let $(P, \bsigma)$ be induced by an iterative partition $P$ computing $\phi$.
    \[
    \Gamma(P, \bsigma) = \{ (\btheta, i) \mid \exists\, \theta_i' : P \text{ distinguishes } \btheta \text{ and } (\theta_i', \btheta_{-i})\text{, and }\phi(\btheta) = \phi(\theta_i', \btheta_{-i})\}.
    \]
\end{lemma}

If the context is clear, we will abuse notation and terminology for ease of exposition. We will call an iterative partition $P$ for $\phi$ a \emph{protocol}, and write $\Gamma(P,\phi)$ to denote the contextual privacy violations of the protocol induced by $P$ and $\phi$. 

In light of \Cref{prop:partitional} and \Cref{lem:partitionalcpviolations}, there is a contextually private protocol for $\phi$ if and only if there is an iterative partition $P$ for $\phi$ such that $\Gamma(P, \phi) = \emptyset$. As another consequence, for any protocol for $\phi$ that is maximally contextually private, there a protocol induced by an iterative partition $P$ for $\phi$ such that $\Gamma(P, \phi)$ is inclusion-minimal within $\{ \Gamma(P', \phi) \mid P' \text{ is an iterative partition for }\phi\}$.

\subsection{Properties of Contextual Privacy}\label{sec:whycp}
 
In this subsection we explain the main modeling choices behind our definition of a contextual privacy violation. For clarity of exposition, we restrict attention to iterative partitions $P$, and we will often use the notation $\Gamma(P,\phi)$ for the set of contextual privacy violations induced by protocol $P$ and choice rule $\phi$. We proceed by unpacking the definition of a contextual privacy violation: it conditions on a social choice function $\phi$, on a type profile $\btheta$, and on a particular agent $i$. Along the way, we will mention other notions of privacy---group contextual privacy and relative informativeness---which are discussed in more detail in \Cref{sec:extensions} and \Cref{sec:informativeness}, respectively.

\emph{Dependence on choice rule $\phi$.} The defining feature of contextual privacy violations is their dependence on the choice rule $\phi$---this notion of a privacy violation is \textit{contextual} in the sense that it depends on what the designer is trying to do with the information they gather. Formally, we only count a distinction between $\theta_i$ and $\theta_i'$ as a violation at $\btheta$ if $\phi(\theta_i, \btheta_{-i}) = \phi(\theta_i', \btheta_{-i})$. 

If we dropped the condition $\phi(\theta_i,\btheta_{-i})=\phi(\theta'_i,\btheta_{-i})$, the only means to compare protocols (and their resulting partitions) is distinction. This leads to the privacy partial order defined on partitions of the type space known as \emph{relative informativeness} (see  \citet{segal2007communication} and \citet{mackenzie2022menu}).

\begin{definition}
    Let $(P, \bsigma)$, $(P', \bsigma')$ be iterative partitions. We say that $(P, \bsigma)$ is less \emph{relatively informative} if for all $\btheta, \btheta' \in \bTheta$:
    \[
    (P, \bsigma) \text{ distinguishes } \btheta, \btheta' \Rightarrow (P', \bsigma') \text{ distinguishes } \btheta, \btheta'.
    \]
\end{definition}
Relative informativeness compares the total information revealed through two protocols, treating all disclosures as equally undesirable. However, in mechanism design, the disclosure of information is instrumental: the designer \emph{must} distinguish type profiles that map to different outcomes to implement $\phi$. A privacy metric that penalizes the designer for learning necessary information may not capture the range of concerns that arise in economic analysis. 

Contextual privacy acknowledges the contextual nature of privacy loss, i.e. the importance of the social choice rule. It imposes no penalty for learning information that is pivotal to the allocation. This captures the idea that privacy loss occurs only when the mechanism reveals information that is \emph{superfluous} to the economic context. This aligns with practical concerns, such as in the 2017 FCC Incentive Auction, where \enquote{unconditional winner privacy} was desired specifically to hide the winner's surplus (the difference between the highest bid and the price), while the information that determined the price was necessary for making the allocation, and thus not a privacy concern \citep{milgrom2020clock}.\footnote{Related contextual concerns arise in practice for loser privacy as well. For instance, a lawsuit brought by the Texas Attorney General against Google Inc.\ alleged that the company stored losing bids in second-price advertising auctions and used this information against advertisers in future interactions \citep{texasvgoogle2022}.}

In capturing superfluous information, contextual privacy is at the level of the agent and state pair: it detects whether a violation occurs for agent $i$ at profile $\btheta$ but not how much excess information is revealed. If at a particular type profile $\btheta$ a protocol $(P, \bsigma)$ reveals more information about agent $i$ than is needed, the notion does not differentiate \textit{how much} superfluous information is revealed. That is, if the protocol distinguishes $\theta_i$ from $\theta_i'$ and they lead to the same outcome holding fixed $\btheta_{-i}$, there is a violation; distinguishing $\theta_i$ from yet another superfluous type $\theta_i''$ (holding fixed $\btheta_{-i}$) does not register as another violation. 

Normatively, this reflects the idea that what matters for agent $i$ in the realized state of the world $\btheta$ is \textit{whether} the mechanism exposes them to \emph{any} unnecessary disclosure relative to the choice rule; once such a disclosure occurs, the core privacy concern for that agent–state pair is already present, and our framework does not, therefore, rank degrees of overexposure there.

For designers who wish to compare the magnitude of revelation among protocols that are contextually privacy-equivalent, \emph{relative informativeness} serves as a natural refinement. We give an example in \Cref{sec:informativeness} showing how relative informativeness and contextual privacy are complementary---relative informativeness can further select among contextual-privacy equivalent protocols. Given the complementary insights that relative informativeness delivers, we also characterize minimally relatively informative protocols (\Cref{prop:relativeinformativeness}) and show that a maximally contextually private protocol of interest in this article, the ascending-join protocol, is also minimally relatively informative (\Cref{prop:informativeness_ascj}). 

\emph{Dependence on type profile $\btheta$.} Second, the definition of a contextual privacy violation is conditioned on a particular type profile: there is a violation \emph{at type profile $\btheta$} if there is an agent whose type is distinguished from another type holding all others fixed, even though the outcome of the choice rule is the same at both types. This state‑by‑state perspective allows us to capture ``local'' improvements in privacy: we can modify a protocol to protect a specific agent at a specific type profile without needing to weigh this protection against potential losses elsewhere in the type space, provided the set of violations shrinks. For example, if two protocols $P$ and $P'$ compute the same social choice function but $P'$ eliminates violations on a subset of type profiles $\btheta$, while producing the same violations elsewhere, then $\Gamma(P',\phi)\subseteq\Gamma(P,\phi)$ and $P'$ is more contextually private. Contrast this with a more ``ex-ante" notion of contextual privacy which might state that protocol $P$ produces a violation if \textit{there exists} a type profile $\btheta$ at which superfluous information is revealed. This would yield a notion that could be well captured by a scalar measure of privacy loss. Others are based on Shannon entropy \citep{eilat2021bayesian} or Kullback-Leibler divergence \citep{liu2020preserving}.  Scalar measures necessarily aggregate privacy losses across states (often using a prior). Consequently, minimizing a scalar measure implies a willingness to trade off a privacy violation at state $\btheta$ against a violation at state $\btheta'$. By avoiding aggregation, our definition induces a dominance relation over protocols: conditional on implementing the same social choice rule, a protocol is maximally contextually private if its set of violations is inclusion-minimal—equivalently, if it lies on the Pareto frontier of implementation and contextual privacy, in the sense that no other implementing protocol weakly reduces every violation and strictly reduces at least one.

\emph{Dependence on agent $i$.} Finally, the definition of a contextual privacy violation assigns violations to \emph{individual agents}. There is a contextual privacy violation for agent $i$ at $\btheta$ if the designer can distinguish $\theta_i$ from some $\theta_i'$ at $\btheta$, holding $\btheta_{-i}$ fixed. Thus $\Gamma(P,\phi)$ keeps track of how much information the protocol elicits \emph{about each agent}. An alternative approach would be to define \emph{group contextual privacy}, where a violation occurs if the designer learns \emph{any} property of the joint type profile that is irrelevant to the outcome.

\begin{definition}
    Let $P$ and $P'$ be iterative partitions for $\phi$. We say that $P$ is more \emph{group contextually private} than $P'$ if
    \begin{multline*}
     \{ \btheta \mid \exists\, \btheta' \colon P \text{ distinguishes } \btheta \text{ and } \btheta' \text{, and }\phi(\btheta) = \phi(\btheta')\} \\
     \subseteq \{ \btheta \mid \exists\, \btheta' \colon P' \text{ distinguishes } \btheta \text{ and } \btheta' \text{, and }\phi(\btheta) = \phi(\btheta')\}.
    \end{multline*}
    We say that $P$ is \emph{fully group contextually private} if $P$ produces no group contextual privacy violations.
\end{definition}

In \Cref{subsec:groupcp} we show that full group contextual privacy is extremely demanding: mechanisms that are fully group contextually private must, in effect, rule out sets of possible outcomes whenever any agent interacts with them (\Cref{prop:groupcp}). This stringency makes it difficult to obtain a rich comparative order: in our iterative-partition formalism, a departure from group contextual privacy occurs exactly at nodes whose queries fail to shrink the associated outcome set, and each such node generates violations for all type profiles that reach it. As a result, improving in the group contextual privacy order means restructuring the entire protocol tree, rather than making state- or agent-specific privacy improvements. Consequently, the induced partial order is very coarse: beyond separating mechanisms that satisfy group contextual privacy at every node from those that do not, it offers little additional structure for comparing mechanisms. For our purposes, it is thus largely sufficient to characterize the fully group contextually private mechanisms, which is exactly what \Cref{prop:groupcp} does.

Furthermore, normative accounts of privacy often center on information about individuals rather than about groups. For example, legal scholar Helen Nissenbaum’s influential theory of \emph{contextual integrity}, which served as loose inspiration for the notion we develop here, describes privacy as a right concerning information about individual persons, and \enquote{does not take up questions about the privacy of groups or institutions} \cite[p.\ 123f.]{nissenbaum2004privacy}. For both theoretical and normative reasons, we therefore adopt an individual-level notion and return to group contextual privacy only as an extension in \Cref{sec:extensions}.

\section{Contextual Privacy and Pivotality} \label{sec:generic}
We first study the social choice rules that meet the strict requirement of \textit{full contextual privacy}. Doing so not only shows us that some well-known choice rules (first-price auction rules, serial dictatorships), admit fully contextually privacy protocols, but also uncovers a relationship between privacy and different forms of pivotality. 

For this analysis, we define individual and collective pivotality. On a product set $\hat \bTheta  = \bigtimes_{i=1}^n \hat \Theta_i \subseteq \bTheta$, an agent $i$ is \emph{individually pivotal} for $\phi$ if there exist $ \emptyset \subset \tilde \Theta_i \subset \hat\Theta_i$ such that for all type profiles $(\theta_i, \btheta_{-i}) \in \tilde \Theta_i\times \hat \bTheta_{-i}$ and $(\theta_i', \btheta_{-i}) \in (\hat \Theta_i \setminus \tilde \Theta_i )\times \hat \bTheta_{-i}$, the outcomes under $\phi$ differ, i.e. 
\[
\phi(\theta_i, \btheta_{-i}) \neq \phi(\theta_i', \btheta_{-i}).
\]
We say that agents $i = 1, 2, \dots, n$ are \emph{collectively pivotal} for $\phi$ on $\hat \bTheta$ if there are $\btheta, \btheta' \in \hat \bTheta$ such that
\[
\phi(\btheta) \neq \phi(\btheta').
\]
Individual pivotality aligns with our usual use of the term in mechanism design and social choice---it captures whether, for a particular subset of possible agent types, changes in the agent's type alter the outcome. Collective pivotality captures whether some group's joint type-profile influence the outcome, holding others not in the group fixed. These notions help us characterize, in the next result, choice rules in which contextual privacy violations arise.
\begin{theorem}
\label{thm:characterization}
A social choice function $\phi \colon \bTheta \to X$ admits a fully contextually private protocol if and only if for all product sets $\hat\bTheta=\bigtimes_{i=1}^n \hat\Theta_i$ on which the agents are collectively pivotal, at least one agent is individually pivotal.
\end{theorem}
\begin{proof}
We first show (in the contrapositive) that a choice rule being fully contextually private implies that no such set $\hat{\bTheta}$ exists. Assume that there is a $\hat\bTheta$ on which the agents are collectively pivotal and no agent is individually pivotal.  Let $P$ be any iterative partition for $\phi$. By collective pivotality, $\phi|_{\hat{\boldsymbol\Theta}}$ is non-constant, and so the set of nodes of $P$ that distinguish two type profiles in $\hat{\boldsymbol{\Theta}}$ is non-empty. Let $v$ be any earliest (i.e. minimal in precedence order) node that distinguishes two type profiles in $\hat{\boldsymbol{\Theta}}$. As no agent is individually pivotal, there exist $\btheta =(\theta_i, \btheta_{-i})$ and $\btheta' = (\theta_i', \btheta_{-i})$ that lie in different children of $v$, yet $\phi(\boldsymbol\theta) = \phi(\boldsymbol \theta')$. These constitute a contextual privacy violation for one of the type profiles $\btheta \in \hat \bTheta$.

For the converse direction, assume that no such $\hat{\bTheta}$ exists. We construct a contextually private protocol inductively. For the base case, we consider $\bTheta$. For the inductive step, we consider an arbitrary node  $\tilde{\bTheta}$. Either, $\phi|_{\tilde{\bTheta}}$ is constant, and the protocol can terminate and compute $\phi$, or not. If $\phi|_{\tilde{\bTheta}}$ is not constant, then there is collective pivotality, and, by assumption, there must be an agent who is individually pivotal, witnessed by $\breve\Theta_i \subseteq \tilde\Theta_i$. We can partition $\tilde{\bTheta}$ into two children $\tilde\bTheta' = \breve \Theta_i \times \tilde\bTheta_{-i}$ and $\tilde\bTheta'' = (\tilde\Theta_i \setminus \breve \Theta_i) \times \tilde\bTheta_{-i}$. By individual pivotality, type profiles distinguished at $v$ lead to different outcomes, and hence, distinguishing these does not introduce contextual privacy violations. By induction, no type profiles that lead to the same outcome are distinguished, yielding $\Gamma (P, \phi) = \emptyset$, and a contextually private iterative partition. 

It remains to show that this induction terminates. As type spaces are finite, there is a finite number of type profiles. Since the cardinality of children $\hat\bTheta', \hat\bTheta''$ is strictly smaller than the cardinality of parent $\hat\bTheta$, the size of $\hat\bTheta$ decreases strictly along paths in the partition, showing that the induction terminates after at most $\lvert \bTheta \rvert$ many steps.
\end{proof}
In general, this characterization suggests that many social choice rules will not be fully contextually private. After all, in many choice rules of interest, there are situations in which no single player’s report can move the outcome, but several reports together can. Intuitively, these are “threshold” or “aggregation” environments where the rule reacts only after enough agents change their reports. For assignment rules, these regions are ubiquitous: with discrete capacities and priority cutoffs, one player’s changed report often leaves all cutoffs unchanged, while a few coordinated changes can shift a cutoff, create or remove a blocking pair, or unlock a trading cycle, changing the allocation. Tie-breaking and priorities amplify this---outcomes are locally insensitive to single reports but jump discretely once enough reports move. We will see in a number of results derived from \Cref{thm:characterization}---\Cref{cor:collectindividual}, \Cref{cor:corners}, \Cref{prop:spanoncp}, \Cref{thm:stabsecp} and more in \Cref{sec:full_cpimpossibilities}---that contextual privacy violations are common in choice rules of interest. 

However, before shifting to the negative results, we show that \Cref{thm:characterization} also delivers some positive results: first-price auction rules and serial dictatorships are fully contextually private. 

    Throughout, when we refer to ``first-price auction rules," we are referring to social choice rules induced by a profile of strictly increasing bidding maps
$\mathbf{b}=(b_i)_{i=1}^n$ that translate types into bids. Fix $b_i:\Theta_i\to\R_+$ strictly increasing for each $i$. The induced social choice rule is
\[
\phi_{\FPA}^{\mathbf{b}}(\boldsymbol\theta)
=
\left(\min\argmax_{i = 1, 2, \dots, n} b_i(\theta_i), \max_{i = 1, 2, \dots, n} b_i(\theta_i)\right),
\]
i.e., the winner is the highest bidder (ties broken by smallest index) and the transfer equals the highest bid.\footnote{Note that our privacy analysis studies the computation of $\phi_{\FPA}^{\mathbf{b}}$ from types and is independent of equilibrium considerations, and in specific settings, particular $\mathbf{b}$ functions may be of particular interest (e.g. in the independent private values model, the well-known symmetric Bayes-Nash equilibrium provides a concrete $\mathbf{b}$ of interest).} Since our analysis relies only on order comparisons of $b_i(\theta_i)$, we may write $\phi_{\FPA}$ for brevity. We order the outcomes $(i, t)$ lexicographically in transfer $t$ and winner $i$.
\begin{proposition}
\label{prop:fbacp}
The first-price auction rule is fully contextually private.
\end{proposition}
\begin{proof}
    Let $\hat \bTheta = \bigtimes_{i=1}^n \hat\Theta_i \subseteq \bTheta$ be a product set on which agents are collectively pivotal. Denote agent $i$ the agents with the highest still possible type, i.e. agents of type $\max \hat \Theta_i$ (with tiebreaking by smallest index). We claim that agent $i$ is individually pivotal for the set $\tilde \Theta_i \coloneqq \{\max \hat \Theta_i \}$. (By collective pivotality, it must be that $\hat \Theta_i \setminus \tilde \Theta_i \neq \emptyset$.) In particular, for all type profiles where $\theta_i \in \tilde \Theta_i$ (that is, $\theta_i = \max \hat \Theta_i$) and any other partial type profiles $\btheta_{-i} \in \hat \bTheta_{-i}$, it must be that agent $i$ wins and pays price $b_i(\theta_i)$. For any type profile such that $\theta_i \in \hat \Theta_i \setminus \tilde \Theta_i$, it must be that $\theta_i < \max \hat \Theta_i$. For any $\btheta_{-i} \in \hat\bTheta_{-i}$, either agent $i$ loses or they pay a lower price, establishing individual pivotality.
\end{proof}
The feature of the first price auction that allows for perfect privacy protection is that, for \emph{any} product set, an individually pivotal agent can be declared: it is the agent that, if they had the highest possible type under the product set, they would win independent of all other types that would lie in the product set. By repeatedly singling out this agent, the dynamic protocol can always find an agent that will not produce a contextual privacy violation.\footnote{Others, using different methods of analysis, have pointed out that the descending auction implementation of the first-price auction has the beneficial property of preventing information leakage \citep{gretschko2014information, kleinberg2016descending}.}

A similar positive result holds for the serial dictatorship choice rule. Consider a finite set $C$ of objects. Outcomes are $X = \{ x \colon N \to C \text{ injective}\}$ and types are $\theta_i \in \R^{\lvert C\rvert}$, for some set of \emph{objects} $C$. We assume strict preferences, that is $\theta_c \neq \theta_{c'}$ for $c \neq c' \in C$. Utilities are $u_i(x ; \theta_i) = (\theta_i)_{x(i)}$.

To define the serial dictatorship, we define a lexicographic order $\preceq_{\operatorname{SD}, \btheta}$ on $X$ with tiebreaking $i = 1, 2, \dots, n$. (We will call a choice rule that is the serial dictatorship choice rule for some tiebreaking order \emph{a} serial dictatorship choice rule.) Define $x \preceq_{\operatorname{SD}, \btheta} x'$ for $x, x' \in X$ if and only if $u_i(x; \theta_i) \le u_i(x'; \theta_i)$ for $i = 1, 2, \dots, j$ and $u_{j+1}(x; \theta_{j+1}) < u_{j+1}(x'; \theta_{j+1})$ for some $j = 1, 2, \dots, n-1$. Then, the serial dictatorship choice rule for tiebreaking according to $i = 1, 2, \dots, n$ is
\[
\phi_{\operatorname{SD}}(\btheta) = \max_{\preceq_{\operatorname{SD}, \btheta}} x.
\]
We can similarly use \Cref{thm:characterization} to show full contextual privacy of serial dictatorship choice rules.
\begin{proposition}\label{prop:sd_cp}
    Serial dictatorship choice rules are fully contextually private.
\end{proposition}
\begin{proof}
    Let $\hat \bTheta = \bigtimes_{i=1}^n \hat\Theta_i \subseteq \bTheta$ be a product set on which agents are collectively pivotal. Denote agent $i$ the highest-priority agent that has at least two possible assignments, i.e. the highest-priority agent for which there exist $x, x' \in \phi_{\operatorname{SD}}(\hat\bTheta)$ with $x(i)\neq x'(i).$ Such an agent must exist by collective pivotality. By choice of $i$, for all $j<i$ we have $y(j)=y'(j)$ for all $y,y'\in\phi_{\operatorname{SD}}(\hat\bTheta)$.  
    
     Define $S_i = \{x(i) \colon x \in \phi_{\operatorname{SD}}(\hat\bTheta)\}$. Choose an object $c \in S_i$. We claim that agent $i$ is individually pivotal for the set $\tilde \Theta_i = \{ \theta_i \in \hat\Theta_i \mid \theta_i(c) > \theta_i (c'), \text{ for all } c' \in S_i \setminus \{c\}\}$. 

    For all type profiles where $\theta_i \in \tilde \Theta_i$ and any other partial type profiles $\btheta_{-i} \in \hat \bTheta_{-i}$, it must be that agent $i$ is assigned $c$. For any type profile such that $\theta_i \in \hat \Theta_i \setminus \tilde \Theta_i$, it must be that $i$ is not assigned to $c$.
\end{proof}
Note, this result---which we derive from our general characterization of fully contextually private choice rules to illuminate the connection between contextual privacy and pivotality---could alternatively be obtained as a corollary of an observation in Appendix B of \citet{mackenzie2022menu}. \citeauthor{mackenzie2022menu} show that the serial‐dictatorship menu mechanism is \emph{minimal in the relative informativeness order}: its observed play distinguishes types only up to the induced outcome. We examine further connections between contextual privacy and relative informativeness in depth in \Cref{sec:informativeness}.

There are two noteworthy necessary conditions which follow from \Cref{thm:characterization}. The first necessary condition, \Cref{cor:collectindividual}, corresponds to product sets $\hat\bTheta$ where agents have at most two possible types holding all other types fixed, $\lvert \hat \Theta_i \rvert \le 2$ for all $i \in N$. The second condition is more restrictive.

\begin{corollary}
 \label{cor:collectindividual}
Let $\phi$ be a social choice function, and consider a type profile $\btheta \in \bTheta$. If there exists a subset $A \subseteq N$ of agents, and types $\theta_i' \in \Theta_i$ for agents $i \in A$, 
\[
\phi(\btheta_A,\btheta_{-A}) \neq \phi(\btheta_{A}', \btheta_{-A}) \tag{collective pivotality}
\]
and for all $i \in A$
\[
\phi(\theta_i,\btheta_{-i}) = \phi(\theta_i', \btheta_{-i}) \tag{no individual pivotality}
\]
then for any iterative partition, there is an agent $i \in A$ whose contextual privacy is violated at $\btheta$, i.e. $A \times \{\btheta\} \cap \Gamma(P, \phi) \neq \emptyset.$
\end{corollary}

The next condition is a more restrictive function, which can be thought of as a ``minimal necessary condition" for full contextual privacy. An analogue of this condition appears in the earlier literature on unconditional privacy, where it is called the Corners Lemma \citep{chor1989zero}.

\begin{corollary}[Corners Lemma]\label{cor:corners}
Let $\phi$ be a social choice function, and consider a type profile $\btheta \in \bTheta$. If there are two agents $i, j \in N$ and types $\theta_i' \in \Theta_i$, $\theta_j' \in \Theta_j$ such that 
\[
\phi(\btheta) = \phi(\theta_i', \btheta_{-i}) = \phi(\theta_j', \btheta_{-j}) = x,\]
yet
\[\phi(\theta_i', \theta_j', \btheta_{-ij}) \neq x,\] 
then for any iterative partition $P$ for $\phi$, there must be a contextual privacy violation for agent $i$ or $j$ at $\btheta$, i.e. $\{i, j\} \times \{\btheta\} \cap \Gamma(P, \phi) \neq \emptyset.$
\end{corollary}

We illustrate the Corner's Lemma in \Cref{fig:corner}. If any three ``corners" in a two-by-two, then the fourth ``corner" must also lead to that outcome, otherwise there is a violation. 

\begin{figure}
    \centering
    \begin{tikzpicture}[scale=0.65]
  \fill[gray!20] (0,0) rectangle (2,2); 
  \fill[gray!20] (0,2) rectangle (2,4); 
  \fill[gray!20] (2,2) rectangle (4,4); 

  \draw (0,0) rectangle (4,4);
  \draw (2,0) -- (2,4);
  \draw (0,2) -- (4,2);

  \node at (1,3) {$x$};
  \node at (3,3) {$x$};
  \node at (1,1) {$x$};
  \node at (3,1) {$x'$};
    \end{tikzpicture}
    \caption{Illustration of the Corners Lemma (\Cref{cor:corners}).}
    \label{fig:corner}
\end{figure}

In general, many choice rules we care about do not have fully contextually private protocols---indeed even a minimal counterexample of the sort shown in \Cref{cor:corners} can be used to show that many choice rules are not fully contextually private. To see this, we consider two common choice rules---from an auction domain and one from a matching domain.

First we consider the ($k$-item) Vickrey auction rule for allocating $k$ items $k=1, \dots, n-1$, which we denote $\phik$. The associated choice rule can be represented as $\phik \colon \bTheta \to 2^N \times \Theta, \btheta \mapsto (W, t)$, where $t(\btheta) = \btheta_{[k+1]}$ is the $(k+1)$st highest bid, and $W(\btheta)$ is a set of $k$ winners minimal in the strong set order on $N$, i.e. tiebreaking in the order $i = 1, 2, \dots, n$.\footnote{Sets $W, W'$ of a partially ordered space are ordered in strong set order if $w \le w'$ for all $w \in W, w' \in W'$.} A special case is $\phispa$, the second-price auction rule with a single item. 

\begin{restatable}{proposition}{spanoncp}
\label{prop:spanoncp}
    Let $\Theta_i = \Theta$ for $i\in N$, $\lvert \Theta \rvert \ge 3$, and $k \le n-2$. Then, the Vickrey auctions $\phik$ induce contextual privacy violations.
\end{restatable}

 \begin{proof}
The proof constructs an instance of collective but not individual pivotality using \Cref{cor:corners}. Consider a type profile with $\theta_{[k]} > \theta_{[k+3]}$ (if $n=k+2$, then let $ \theta_{[k+3]}$ be strictly smaller than all $\theta \in \Theta$). 

Consider agents $i, j$ that have types in $[\underline{\theta}, \bar{\theta}]$ such that $\btheta_{[k+4]} < \underline{\theta}<\bar{\theta} < \btheta_{[k+1]}$. Consider the product set $\hat \bTheta \coloneqq \{\underline{\theta},\bar{\theta}\} \times \{\underline{\theta},\bar{\theta}\} \times \bigtimes_{\ell \in N \setminus \{i,j\}} \Theta_\ell$. This corresponds to a square depicted in \Cref{fig:square_spa}.

\begin{figure}[!ht]
    \centering
  \begin{tikzpicture}[scale=.65]

\draw (0,0) rectangle (4,4);
\draw (2,0) -- (2,4);   
\draw (0,2) -- (4,2);   

\foreach \y in {0.5,1.5,2.5,3.5}
  \draw[dashed,gray] (0,\y) -- (4,\y);

\fill (0.5,3.5) circle (1.5pt);
\draw[line width=0.8pt] (1.1,3.5) -- (1.7,3.5);
\draw[line width=0.8pt] (3.1,3.5) -- (3.7,3.5);

\fill (3.0,2.5) circle (1.5pt);

\fill (0.5,1.5) circle (1.5pt);

\draw[line width=0.8pt] (1.1,0.5) -- (1.7,0.5);
\fill (2.7,0.5) circle (1.5pt);
\draw[line width=0.8pt] (3.1,0.5) -- (3.7,0.5);

\node[left] at (-0.2,3.5) {$\bar{\theta}$};
\node[left] at (-0.2,2.5) {$\underline{\theta}$};
\node[left] at (-0.2,1.5) {$\bar{\theta}$};
\node[left] at (-0.2,0.5) {$\underline{\theta}$};

\begin{scope}[xshift=6cm]
  \fill[gray!20] (0,0) rectangle (2,2); 
  \fill[gray!20] (0,2) rectangle (2,4); 
  \fill[gray!20] (2,2) rectangle (4,4); 

  \draw (0,0) rectangle (4,4);
  \draw (2,0) -- (2,4);
  \draw (0,2) -- (4,2);

  \node at (1,3) {$x$};
  \node at (3,3) {$x$};
  \node at (1,1) {$x$};
  \node at (3,1) {$x'$};
\end{scope}

\end{tikzpicture}
    \caption{Showing collective but not individual pivotality to the second-price auction. Type profiles for agent $i$ and agent $j$ (left, agent $j$'s type is represented by a dot, and agent $i$'s type is represented by a dash); required outcome under the second-price auction rule $\phispa$. }
    \label{fig:square_spa}
\end{figure}

The set of winners under $\phik$ is the same in this set, but the price differs. Let $x$ be the outcome in which winners pay $t = \overline{\theta}$ and $x'$ be the outcome where they pay $t = \underline{\theta}$. The agents are collectively pivotal on $\hat \bTheta$. But they are not individually pivotal, as 
\[
\phik(\overline{\theta}, \overline{\theta}, \boldsymbol\theta_{-ij}) = \phik(\underline{\theta}, \bar{\theta}, \btheta_{-ij}) = \phik(\bar{\theta}, \underline{\theta}, \btheta_{-ij}) =x
\]
and
\[
\phik(\underline{\theta}, \underline{\theta}, \btheta_{-ij}) = x'.
\]
\end{proof}

A similar application of \Cref{cor:corners} rules out full contextual privacy of stable rules in two-sided matching contexts. We define a two-sided matching problem as follows. Every agent (also called \enquote{student}) is matched to at most one object (also called \enquote{school}), and at most $\kappa(c)$ agents are matched to an object $c$, for every $c \in C$, for some \emph{capacities} $\kappa(c)$. That is, the set of outcomes is
\begin{multline*}
X = \left\{ 
\mu \subseteq N \times C \colon\right. \left.\forall i \in N\colon  \lvert \{c \in C \mid (i, c) \in \mu \} \rvert \le 1  \text{ and }  \forall c \in C \lvert \{i \in N \colon (i, c) \in \mu\} \rvert \le \kappa(c)\right\}.
\end{multline*}
We say there is \emph{no oversupply} if the aggregate capacity equals the number of agents, $\sum_{c \in C} \kappa(c) = n$. 

We assume that objects' preferences over agents are determined by agents' \emph{priority scores}. We assume the scores for different objects are private information, and so an agents ``type" consists both of her score vectors at different objects and her preference ordering  over the objects. This matches the \emph{college assignment problem} with standardized test scores \citep{balinski1999tale}. We assume that each agent has a vector of \emph{scores} $(\score_c)_{c \in C}$, representing their score at each object $c \in C$. Objects prefer agents with higher scores. Agent $i$ has private information $\theta_i = (\preceq_i, \score_i)$, where $\preceq_i$ is $i$'s preference ranking over schools, and $\score_i \colon C \to \R$ maps objects to scores. 

In such school choice settings, a desirable property of choice rules is \emph{stability} (or  \emph{no justified envy}). A choice rule $\phi$ is \emph{stable} or \emph{induces no justified envy} if there is no blocking pair $(i,c)$, $i \in N$, $c \in C$ such that $c \succ_{i} \phi_i(\boldsymbol\theta)$ and $\score_{i}(c) >\score_{j} (c)$ for $\phi_j(\boldsymbol\theta) = c$.

\begin{restatable}{proposition}{stabsecp}
\label{thm:stabsecp}
    Assume $\lvert C\rvert\geq 2$ and $\lvert N\rvert \ge 2$ and that there is no oversupply. Then any iterative partition $P$ for a stable choice rule produces contextual privacy violations.
\end{restatable}

This proposition, which we again prove through the minimal counter-example delivered in \Cref{cor:corners}, highlights one way in which stability produces contextual privacy violations. Whether an agent has justified envy of another student is a collective feature of both agents' types. A single student's change in score may lead to a claim of another student to their seat in this school, yet no justified envy. A change in this other student's preference may lead to envy which is not justified. Yet both changes together lead to justified envy, and force collective pivotality in these changes in type. 

We can also see, for the school choice problem, how common protocols will produce contextual privacy violations. Consider for example the deferred acceptance protocol, which produces a stable outcome. The designer must learn information about tentative assignments that may not be final, and in so doing violates contextual privacy. 

In \Cref{sec:full_cpimpossibilities}, we show additionally that no individually rational and efficient house assignment rule is contextually private (\Cref{subsec:houseassignment}) and that no generalized median-voting rule is contextually private (\Cref{subsec:voting}). 

When a choice rule cannot be implemented by a fully contextually private protocol, the preceding results offer only limited design guidance. We now turn to the problem of finding privacy-conscious protocols when the social choice function is not fully contextually private.

\section{Maximal Contextual Privacy in Auctions}\label{sec:maxcp}
In the previous section, we studied whether contextual privacy violations are unavoidable. Here, we take the perspective of a privacy-conscious designer who must select a protocol for a fixed choice rule. If full contextual privacy of a protocol is unattainable---which is often the case (see \Cref{sec:full_cpimpossibilities})---such a designer must make choices about whose privacy to protect, and at what type profiles. They would like to find protocols that are maximally contextually private: protocols that implement the choice rule, and are maximal elements in the contextual privacy order, i.e. that induce a minimal set (in the inclusion sense) of privacy violations. 

Studying maximal contextual privacy in a high degree of generality is challenging. Fortunately, our results in \Cref{sec:generic} point to structure that can be imposed on the set of choice rules we study, simplifying our analysis. This structure will be most natural for one-dimensional type spaces, leading us to consider auction domains. In particular, we will find in \Cref{subsec:bimonotonicity} a sufficient condition related to pivotality that reduces the set of protocols that must be considered in the search for maximally contextually private ones.

Then, in \Cref{subsec:kpa}, we study maximally contextually private protocols for $k$-item Vickrey auction rules and provide intuition for how our method improves privacy in other market design settings.\footnote{While the main text focuses on the ascending-join protocol, we study in \Cref{sec:overdesc} an analogous maximally contextually private descending protocol, the overdescending-join protocol.} 

\subsection{Bimonotonicity and the Design Dimensions of Contextual Privacy}\label{subsec:bimonotonicity}

First, we show that it is without loss, for a certain class of choice rules, to restrict attention to partitions of a particular structure. In particular, we will be able to make this reduction for choice rules that satisfy \textit{interval pivotality}.  
\begin{definition}[Interval Pivotality]
    A social choice function $\phi$ defined on an ordered type space $\Theta$ exhibits \emph{interval pivotality} if for all $i \in N$, $\btheta_{-i} \in \bTheta_{-i}$, there are elements $\underline{\theta} \in \Theta$ and $\overline{\theta} \in \Theta$ such that for $\theta_i \neq \theta_i'$,
    \[
        \phi(\theta_i, \btheta_{-i}) = \phi(\theta_{i}', \btheta_{-i}) \iff (\theta_i, \theta_{i}' \le \underline{\theta} \text{ or }\theta_i, \theta_{i}' > \overline{\theta}).
    \]
\end{definition}
Many standard auction rules are interval pivotal. To see why, note that in auctions very high types cannot change the outcome (because they are winners) and very low types do not affect the outcome (because they are losers), intermediate types determine the price and hence the marginal function $\theta_i \mapsto \phi(\theta_i, \btheta_{-i})$ is injective. For instance, $\phi_{\FPA}$ is interval pivotal, with $\overline{\theta}(\btheta_{-i}) = \max \Theta_i$ and $\underline{\theta} (\btheta_{-i}) = \max \btheta_{-i}$.

Most relevant for this section, $k$-item Vickrey auction rules satisfy interval pivotality. We present a lemma proving as much here, both to spell out an example of interval pivotality, and because it will be used in \Cref{subsec:kpa}.

\begin{lemma}
\label{lem:kpaintervalpivotal}
    The $k$-item Vickrey auctions $\phik$, $k= 1, 2, \dots, n$ are interval pivotal. 
\end{lemma}
\begin{proof}
    Consider an agent $i \in N$ and $\btheta_{-i} \in \bTheta_{-i}$. We choose 
    \begin{align*}
        \overline{\theta} & \coloneqq (\btheta_{-i})_{[k]} &
        \underline{\theta} & \coloneqq (\btheta_{-i})_{[k+1]}. 
    \end{align*}    
    where $(\btheta_{-i})_{[l]} > \theta$ for $l > n-1$. We consider three cases. First, if $\theta_i >  (\btheta_{-i})_{[k]}$, then $i$ is a winner under $\btheta$, and their type does not change the outcome for any such type. If  $\theta_i \le  (\btheta_{-i})_{[k+1]}$, then agent $i$ is a loser that is not pivotal for the price, and hence no type of these changes the outcome. For types $\underline{\theta} < \theta_i \le \overline{\theta}$, the agent determines the price, and they are changing the outcome, implying interval pivotality.
\end{proof}

Recall that in \Cref{sec:generic}, we demonstrated an intimate connection between contextual privacy and pivotality. In this section, we find new connections between the two concepts. In particular, the pivotality patterns implied by interval pivotality allow us to reduce the space of protocols we consider in the search for maximally contextually private protocols. This sufficient class is \textit{bimonotonic} protocols. 

These protocols restrict the sequence of queries to the same agent. Some preliminary definitions are needed. We say that an iterative partition is \emph{monotonically increasing} for agent $i$ for all $\tilde \bTheta$ such that $i(\tilde \bTheta) = i$, the children of $\tilde \bTheta$ are $\{ \min \tilde \Theta_i\} \times \tilde\bTheta_{-i}$ and $(\min \tilde \Theta_i, \max \Theta_i]\times \bTheta_{-i}$. That is, the agent is repeatedly queried on whether their type is the lowest ``remaining" type. Similarly, we define an iterative partition to be \emph{monotonically decreasing} for agent $i$ if for every $\tilde \bTheta \subseteq \bTheta$ such that $i(\tilde \bTheta) = i$, the children of $\tilde \bTheta$ are $\{ \max \tilde \Theta_i\} \times \bTheta_{-i}$ and $[\min \Theta_i, \max \tilde{\Theta}_i) \times \bTheta_{-i}$.

We call an iterative partition \emph{bimonotonic} if the first query to each agent is a \emph{threshold query} separating $[\min \Theta_i, \tilde \theta]$ from $(\tilde \theta, \max \Theta_i]$ for some $\tilde \theta \in \Theta_i$, and monotonically increasing for agent $i$ on all $\tilde \bTheta \subseteq \bTheta$ that do not contain $\min \Theta_i$ and monotonically decreasing for agent $i$ on all $\tilde \bTheta \subseteq \bTheta$ that do not contain $\max \Theta_i$. Many auction protocols are bimonotonic, including the ascending protocol, the descending protocol, and the overdescending protocol which is discussed in \Cref{sec:overdesc}. There are also some auction protocols rarely seen in practice---which, for instance, begin in the middle of the type space, or at different thresholds for different agents---that are also bimonotonic.\footnote{There is evidence that more exotic bimonotonic protocols that combine descending and ascending sequences of queries were in use in the \nth{17} century in England \citep[p. 22]{Learmount1985history}.}

The following result states that any set of contextual privacy violations that can be achieved with a generic protocol $P$ can also be achieved using a bimonotonic protocol $P'$. The significance of this result, for design, is that a privacy-conscious designer need only consider the set of bimonotonic protocols. 
\begin{restatable}{proposition}{bimon}
\label{thm:bimon}
    Let $\phi$ be an interval pivotal social choice rule. For every iterative partition $P$ for $\phi$, there is a contextual privacy-equivalent bimonotonic iterative partition $P'$ for $\phi$, i.e. $\Gamma(P, \phi) = \Gamma(P', \phi)$.
\end{restatable}
Bimonotonicity offers insight into the structure of design choices when designing with contextual privacy in mind. In particular, the designer makes two choices that affect the set of privacy violations: (i) the initial threshold query to each agent, which we will sometimes call a ``guess," and (ii) the order in which agents are queried. 

Guessing by itself is quite powerful.
\begin{example}
To see this, consider the $k$-item Vickrey auctions $\phik$, $k = 1, 2, 3, \dots, n$, and any outcome $(W, t) \in X$. A partition $P$, which partitions in any order
\[
\begin{cases}
\{[\min \Theta_i , t), [t, \max \Theta_i] \} & i \in W\\
\{[\min \Theta_i, t], (t , \max \Theta_i]\}& i \in N \setminus W
\end{cases}
\]
will yield no privacy violation for any type profile $\btheta$ with $\phi(\btheta) = (W, t)$ and $\theta_i \neq t$ for all $i \in N$. This protocol \enquote{guesses} that the outcome is $(W,t)$ and asks agents whether they have information that rules out this outcome. For the $k$-item Vickrey auction, such guessing is sufficient to compute the choice rule.
\end{example}

That said, such protocols are of limited use in many settings: After all, mechanism design is relevant because the outcome is \emph{not} known.

\subsection{A Maximally Contextually Private Protocol}\label{subsec:kpa}
Careful construction of the order in which agents are queried---the second dimension of design---can lead to contextual privacy improvements that do not depend on knowledge of the type distribution.  We will exposit such improvement in the following example, which takes the ascending protocol as a starting point, and in \Cref{thm:maxcpkpa}. We show analogous improvement of the descending protocol in \Cref{sec:overdesc}.

\begin{example}[Privacy Improvement of the Ascending Protocol]

Consider implementing a $4$-bidder single-item, second-price auction with type space $\Theta=\{1, 2, \dots, 10\}$ with the ascending protocol\footnote{The ascending protocol is the bimonotonic protocol that has $\min \Theta_i$ as an initial guess for all agents and queries agents $i = 1, \dots, n$ sequentially for every type $\theta \in \Theta$.} for the type profile $\btheta=(4,3,8,2)$.

For this type profile, the winner is agent 3, who has value $\theta_3 = 8.$ First, consider the ascending protocol, depicted on the left of \Cref{fig:ascjoinvsasc}.

\begin{figure}[h!]
    \centering
    \includegraphics[scale=.17]{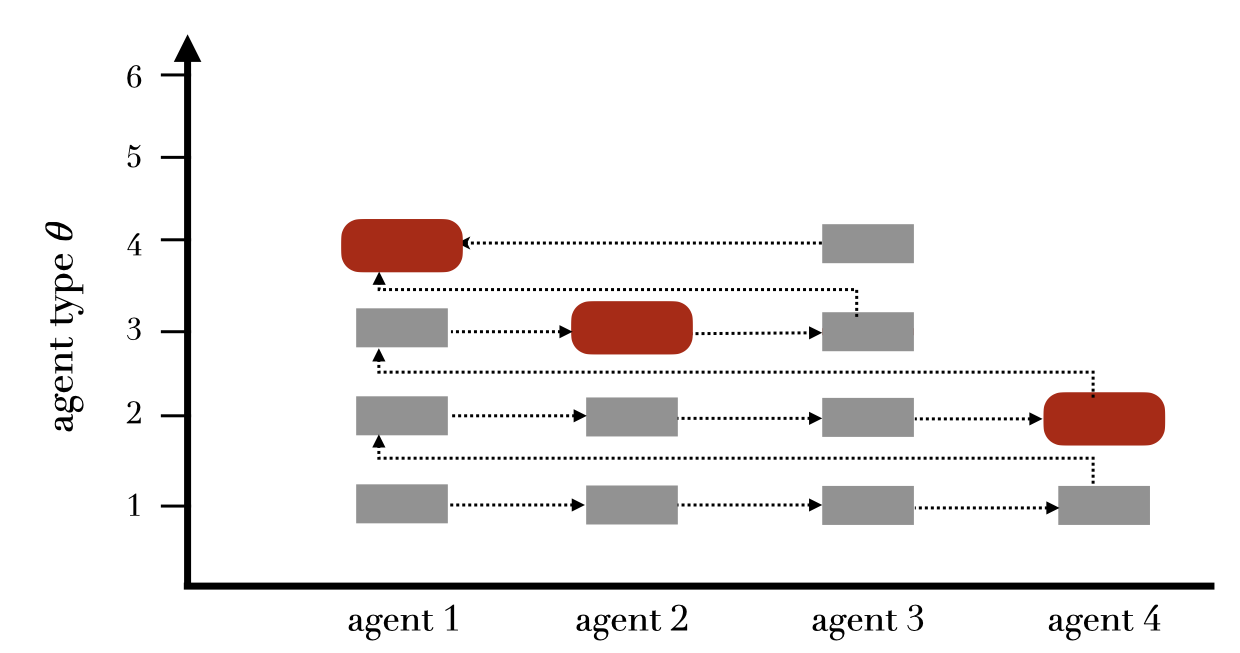}\hspace{1mm}\includegraphics[scale=.17]{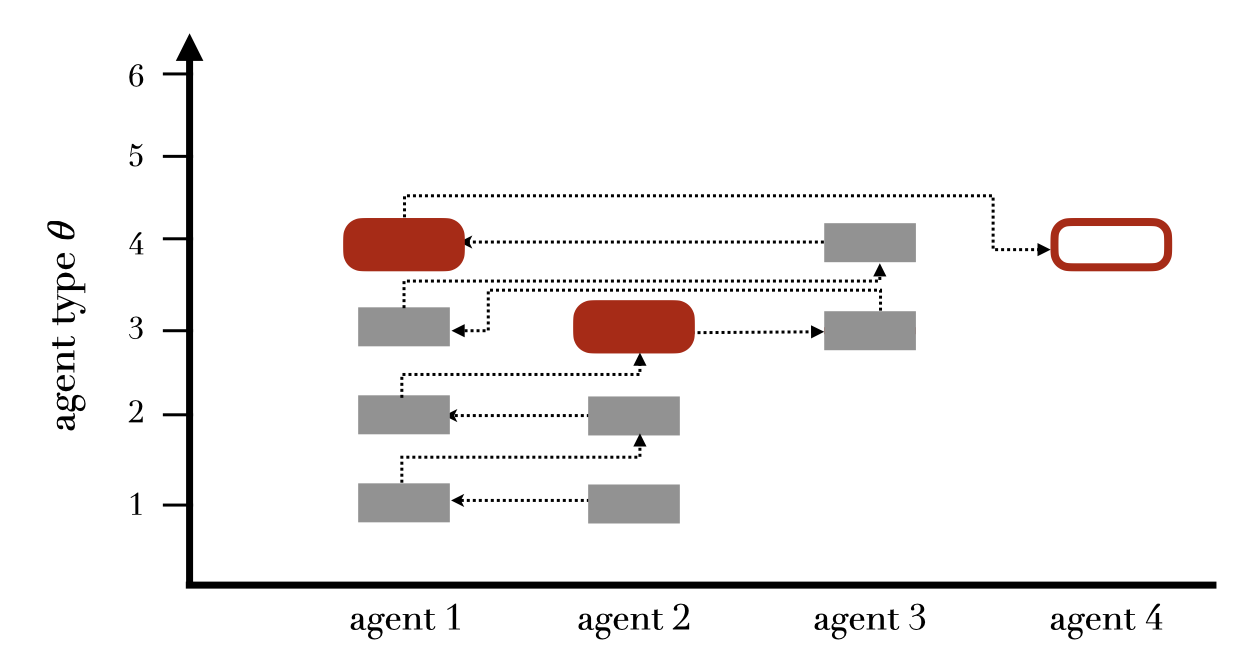} 
    \caption{Example---Ascending (left) vs. Ascending-Join (right) Protocol for $\btheta= (4,3,8,2)$. Queries to an agent are represented as rectangles, positioned at the corresponding threshold $\tilde\theta$ (gray solid rectangles represent affirmative responses, red rounded rectangles represent negative responses that signify the agents' type has been exactly revealed, red unfilled rectangles represent negative responses that do not exactly reveal the agent's type. The path of the dotted line represents the order in which queries are asked. In both protocols, the winner (agent 3) does not have their privacy violated. In the ascending-join protocol, the late-queried loser (agent 4) also does not have their privacy violated. \Cref{thm:maxcpkpa} shows that the ascending-join protocol is maximally contextually private.}
    \label{fig:ascjoinvsasc}
\end{figure}

Note that the protocol determines the exact type of losing bidders who don't determine the price. This information is superfluous hence leading to contextual privacy violations at $(\btheta, 2),$ and $(\btheta, 4)$.

Consider in comparison a protocol that works as follows. Query two agents at a time. Whenever an agent drops out replace them by another agent and continue at the running price (see \Cref{fig:ascjoinvsasc}, right, for illustration). In this protocol, agent 4's privacy can be protected, as they are only queried when the price $\btheta_{[2]}$ has already been found: they are a \emph{late-queried loser}. Note that it is not necessary to query agent $4$ earlier in the protocol, because agents 1 and 2's information is sufficient to \enquote{rule out} prices 1 and 2 without involving agent $4$. The so-constructed protocol still computes the single-item second-price auction choice rule, but produces a strict subset of contextual privacy violations, making it a privacy improvement over the ascending protocol.
\end{example}

The protocol given in the example above can be viewed as a \enquote{join} protocol: Agents are only involved if their information is necessary to increase the running price. For a $k$-item auction, only $k+1$ bidders are necessary to rule out a running price. Therefore, we will call the general protocols underlying the example \textit{ascending-join protocols.}

\begin{definition}
The iterative partition of an \emph{ascending-join protocol} is defined inductively. Starting from the root $\bTheta$, for each $\tilde \bTheta$, denote $x \coloneqq \min \phi(\tilde \bTheta)$ the \emph{tentative outcome}.\footnote{For our main case of $x = (W, t)$, we order outcomes lexicographically by first the transfer $t$ and then the strong set order on the set of winners.} For every $\tilde \bTheta$, query agents sequentially in the lexicographic order $i=1, 2, \dots , n$ with a question equivalent to (i.e. results in the same partition as) one that asks: {can rule out outcome $x$ with your information alone}? To be precise: 
\begin{equation}\label{eq:ruleout}
\begin{split}
\Big\{\{ \tilde\theta_i \in \tilde \Theta_i \mid \forall\, \tilde \btheta_{-i} \in \tilde \bTheta_{-i}:  \phi(\theta_i, \btheta_{-i}') \neq x  \}&,  \{ \tilde\theta_i \in \tilde \Theta_i \mid \exists\, \tilde \btheta_{-i} \in \tilde \bTheta_{-i} : \phi(\theta_i, \tilde\btheta_{-i}) = x \} \Big\} \times \tilde\bTheta_{-i}.\\
\text{\enquote{rule out $x$}} \qquad \qquad  & \qquad \qquad \text{\enquote{cannot rule out $x$}}
\end{split}
\end{equation}
If an agent's response rules out an outcome, the protocol starts with agent $1$ at a new lowest (still-possible) tentative outcome. If agent $n$ cannot rule out the outcome, the protocol terminates and allocates.
\end{definition}
Note that while we define this protocol in terms of queries that are partitionally equivalent to asking whether an agent can ``rule out" a tentative outcome given their information and the information that has been revealed so far, the actual protocol (and strategies) need not ask these questions exactly (as many queries might be trivial). These questions in \eqref{eq:ruleout} simply illustrate how the iterative partition is constructed.

Consider again the alternative protocol presented in the example above: it is an ascending-join protocol implemented through simple queries of the form ``Is your type above $\tilde \theta$?" But to see that it is an ascending-join protocol, note that we can interpret the initial query to agent $1$ as asking whether agent 1 can rule out the outcome $(\{1\}, 1)$ (agent 1 wins at price 1). This is trivial: For no value of $\theta_1$ can agent $1$ rule out $(\{1\}, 1)$. Next, agent 2's response rules out outcome $(\{1\},1)$ if $\theta_2 > 1$. As this changes the set $\tilde \bTheta$, the query begins again with agent $1$. Agent $1$ is asked a question whose response can rule out the next smallest outcome $(\{2\},1)$, which they can, as their type is higher than $1$ as well. Their response can also rule out $(\{3\}, 1)$ and $(\{4\}, 1)$. Agent $1$'s response cannot rule out $(\{1\}, 2)$ given their information alone. Hence agent $2$ is queried, and their response can rule out this type profile. The protocol continues, leading to the partition depicted in \Cref{fig:ascjoinvsasc}, right.

The main modification of the ascending-join protocol, as compared to an ascending protocol, is that it considers a tentative outcome rather than a running price, and asks agents questions whose responses can rule out this outcome. 

The ascending-join protocol not only improves on the contextual privacy of the ascending protocol, but it is also unimprovable, i.e. it is maximally contextually private. 

\begin{restatable}[Maximal Contextual Privacy of Ascending-Join Protocols]{theorem}{maxcpkpa}
\label{thm:maxcpkpa}
The ascending-join protocol is maximally contextually private for $\phik$ and has obviously dominant strategies. 
\end{restatable}

The second sentence follows from well-known properties of personal-clock auctions \citep{li2017obviously}. We provide in an appendix an analogous result for a descending protocol for the $k$-item Vickrey auction which is analogous to the overdescending protocol defined in \citet{overdescending}.

The proof considers any protocol $P$ that is a weak contextual privacy improvement of $P_{\ASCJ}$,  $\Gamma(P, \phik) \subseteq \Gamma(P_{\ASCJ}, \phik)$. To this end, we first characterize the set of contextual privacy violations that arise in the ascending-join protocol, $\Gamma(P_{\ASCJ}, \phik)$ (\Cref{prop:delay}). Then, by induction over nodes $\tilde \bTheta \subseteq \bTheta$, we show that any such $P$ that computes the choice rule needs to ask all queries that $P_{\ASCJ}$ does, and hence $\Gamma(P, \phik) = \Gamma(P_{\ASCJ}, \phik)$, and $P_{\ASCJ}$ is maximally contextually private.

The ascending-join protocol $P_{\ASCJ}$ protects some agents' contextual privacy by delaying questions to them. In particular, it offers contextual privacy protection of some losers, in addition to the winners. To formalize whose privacy is protected, we define the set of late-queried losers. Let $d(\btheta)$ be the index of the $(k+1)$st highest bidding agent, with ties broken in the order $1, 2, \dots, n$. The set of \emph{late-queried losers} is
\[
    \LQL (\btheta) = \{ d(\btheta) + 1, d(\btheta) + 2, \dots, n\} \setminus \W (\btheta).
\]
With this definition, we have the following result.
\begin{restatable}[Delay Protects Contextual Privacy]{proposition}{delay}
\label{prop:delay}
    Let $\btheta \in \bTheta$. The ascending-join protocol delays queries to those it protected; it protects winning and late-queried losing bidders. That is $(\btheta, i) \notin \Gamma (\phik, P_{\ASCJ}) \iff i \in   W(\btheta) \cup \LQL (\btheta)$.
\end{restatable}
\begin{proof}
    Consider $\btheta \in \bTheta$ and consider winners, late-queried losers, and all other agents separately. For winners, only the fact that they are winners is learned in an ascending-join protocol, so they do not have a contextual privacy violation. Late-queried losers, for any node $\tilde \bTheta \subseteq \bTheta$ such that the running price, i.e. $t$ for $(W,t) = \min \phi(\tilde \bTheta)$ is smaller than $\btheta_{[k+1]}$, cannot rule out an outcome before an earlier-queried agent rules it out, \emph{for any possible type}, hence there are no contextual privacy violations at these nodes. As soon as the running price reaches $\btheta_{[k+1]}$, these agents could rule out $\min \phi(\tilde \bTheta)$, but do not (for if they did, they would be, in fact, winners). In these cases, the only information that is learned is that they are losers, $\theta_i \le \btheta_{[k+1]}$, which means that their privacy is not violated. All other agents will rule out an outcome with a price lower than $\btheta_{[k+1]}$, which provides the designer with information $\theta_i \le t < \btheta_{[k+1]}$, in excess of what is needed to know to compute $\phik$. Hence, such agents will incur a privacy violation at $\btheta$.
\end{proof}
In particular, delaying (non-trivial) queries to agents preserves their contextual privacy as much as possible in this maximally contextually private protocol. We show in \Cref{sec:informativeness} that the ascending-join protocol is also a minimal element in the relative informativeness order studied in \citet{mackenzie2022menu}.

Agents who are queried late have fewer contextual privacy violations, as they are asked whether they can rule out an outcome \emph{only if} no earlier-queried agent can rule it out. In other words, by delaying these queries, less information is elicited from these agents, and their contextual privacy is preserved. This has implications for design---if there is some reason why some bidders have values that are more important to protect, then they can be queried late in the procedure.

Moreover, our analysis highlights that every mechanism---whether dynamic or static---implicitly involves a choice about which agents' privacy to protect. Our analysis in this section makes explicit these implicit choices for canonical auction formats---like clock auction protocols. In clock auctions, the set of protected agents is either the winner and some of the losers (in ascending formats) or the losers and some of the winners (in descending formats). We show that a maximally contextually private protocol, the ascending-join protocol, is a particular kind of clock auction where some agents whose privacy is protected are queried late in the process. 

Although the insights in this section were focused on auctions, they do offer intuition for how to  design for contextual privacy in non-auction domains. For example, consider in a school choice problem (compare \Cref{thm:stabsecp}), a modified deferred acceptance protocol $(P_{\operatorname{DA}'}, \bsigma_{\operatorname{DA}'})$ which works as follows: initially exclude a student $i$, compute the deferred acceptance outcome for all other students, and query student $i$ at the very end for the first college whose implied score cutoff they clear. This rule does not lead to any privacy violations for agent $i$---it is more contextually private than the usual deferred acceptance algorithm---so delay provides privacy protection even beyond auction domains.

\section{Variations on Contextual Privacy}\label{sec:extensions}

In this section, we pick up the discussion started in \Cref{sec:whycp} and consider two variations on our definition of contextual privacy violations. The definition of a contextual privacy violation has a component that determines what it means for two types for an agent to be \textit{distinguished} by the designer, and a component that determines whether there is some difference to the outcome that would \textit{justify} the distinction. The variations we consider offer different notions of \textit{distinction} and \textit{justification}, respectively. 

We explore these concepts for both theoretical and practical reasons. On the practical side, these extensions may map onto design goals in some settings. On the theoretical side, these criteria help to illuminate connections to other desiderata in mechanism design, and illustrate which of our results are robust to alternative formulations of contextual privacy.

\subsection{Individual Contextual Privacy}\label{sec:indcp}
The first variation, the \emph{individual contextual privacy violation}, requires that if two types are distinguishable for agent $i$, these types must lead to different outcomes under $\phi$ \emph{for agent} $i$. This definition thus only applies in domains where the outcome space $X,$ specifies an allocation for each agent $i\in N$. Let $\phi_i(\boldsymbol \theta)$ denote the outcome under $\phi$ for agent $i$.  

\begin{definition}[Individual Contextual Privacy Violations]
Let $(P, \bsigma)$ be a protocol for $\phi$. Then the set of \emph{individual contextual privacy violations} $\Gamma_{\operatorname{ind.}} (P, \bsigma, \phi) \subseteq N \times \bTheta$ contains tuples $(i, \btheta)$ for which there exists $\theta_i' \in \Theta$ such that 
\[
\text{$P$ distinguishes $(\theta_i, \btheta_{-i})$ and $(\theta_i', \btheta_{-i})$ yet $\phi_i(\theta_i, \btheta_{-i}) = \phi_i(\theta_i', \btheta_{-i})$.}
\]
We call a choice rule \emph{individually contextually private} for $\phi$ if there is a protocol $(P, \bsigma)$ such that $\Gamma_{\operatorname{ind.}}(P,\bsigma, \phi) = \emptyset$.
\end{definition}
As for the usual set of contextual privacy violations $\Gamma(P, \bsigma, \phi)$, we may restrict to partitional protocols and hence study $\Gamma (P, \phi)$. Notice that individual contextual privacy is stronger than contextual privacy---any choice rule that is individually contextually private is also contextually private. If there were an agent $i$ for whom contextual privacy were violated at $\btheta$, then individual contextual privacy would also be violated for $i$ at $\btheta$.

As a normative criterion, individual contextual privacy requires that if the designer can distinguish between two types for agent $i$, then it \emph{should} be the case that agent $i$'s outcome is changed. This criterion captures a notion of legitimacy---agent $i$ may view participation in the mechanism as involving an inherent tradeoff between information revelation and allocation. We can imagine a speech from agent $i$ along the following lines: \enquote{The designer can learn that I have type $\theta_i$ and not $\theta_i'$ as long as the designer's knowledge of this makes a difference to my allocation.} 

Individual contextual privacy is closely related to \emph{non-bossiness}, introduced by \citet{satterthwaite1981strategy}. A choice rule $\phi$ is \emph{non-bossy} if for all $\theta_i, \theta_i' \in \Theta, \btheta_{-i} \in \bTheta_{-i}$,
\begin{equation}
\phi_i(\theta_i, \boldsymbol \theta_{-i}) = \phi_i(\theta_{i}', \boldsymbol \theta_{-i}) \implies \phi(\theta_i, \boldsymbol \theta_{-i}) = \phi(\theta_{i}', \boldsymbol \theta_{-i}).\label{eq:nonbossy}
\end{equation}
We define the set of \emph{non-bossiness violations} $B(P, \phi) \subseteq N \times \bTheta$ consisting of $(i, \btheta)$ where \eqref{eq:nonbossy} fails to hold.

Non-bossiness says that if agent $i$ changes her report from $\theta_i$ to $\theta_i'$ and her allocation is unchanged, then no other agent $j$'s allocation changes either. The idea is that if agent $i$ could unilaterally change her report and affect a change in some agent $j$'s allocation without changing her own allocation, agent $i$ would be \enquote{bossy.}

We can characterize the set of individual contextual privacy violations $\Gamma_{\text{ind.}}$ in terms of contextual privacy violations and bossiness violations. 

\begin{proposition}\label{lem:icp_nonbossy}
The set of individual contextual privacy violations is the union of contextual privacy violations and non-bossiness violations, $\Gamma_{\text{ind.}} (P, \phi) = \Gamma(P, \phi) \cup B(P, \phi)$.
\end{proposition}
\begin{proof}
We first show that $B(P, \phi) \subseteq \Gamma_{\text{ind.}} (P, \phi)$. Let $(i, \btheta)$ be a non-bossiness violation. Hence, there exists a $j \in N \setminus \{i\}$ and type profiles $(\theta_i, \boldsymbol \theta_{-i})$, $(\theta_i', \boldsymbol \theta_{-i})$ such that $\phi_i(\theta_i, \boldsymbol \theta_{-i}) = \phi_i(\theta_{i}', \boldsymbol \theta_{-i})$ but $\phi_j(\theta_i, \boldsymbol \theta_{-i}) \neq \phi_j(\theta_{i}', \boldsymbol \theta_{-i}).$ Because of the latter, the protocol $P$ must distinguish $(\theta_i, \boldsymbol \theta_{-i})$ and $(\theta_i', \boldsymbol \theta_{-i})$. By the former property, this leads to an individual contextual privacy violation for agent $i$. By definition, $\Gamma(P, \phi) \subseteq \Gamma_{\text{ind.}} (P, \phi)$. Hence, $B(P, \phi) \cup \Gamma(P, \phi) \subseteq \Gamma_{\text{ind.}}(P, \phi).$

To show the converse, we show that $(N \times \bTheta) \setminus (\Gamma(P, \phi) \cup B(P, \phi)) = ((N \times \bTheta) \setminus \Gamma(P, \phi)) \cap ((N \times \bTheta)\setminus B(P, \phi)) \subseteq (N \times \bTheta) \setminus \Gamma_{\text{ind.}}(P, \phi)$, which finishes the proof. In other words, we show that if at $(i, \btheta)$ there is neither a contextual privacy nor a non-bossiness violation, there can't be an individual contextual privacy violation. Consider any $(\theta_i, \btheta_{-i})$ and $(\theta_i', \btheta_{-i})$ that $P$ distinguishes. By contextual privacy,
$\phi(\theta_i, \boldsymbol\theta_{-i}) \neq \phi (\theta_i', \boldsymbol\theta_{-i})$. By non-bossiness, $\phi_i(\theta_i, \boldsymbol\theta_{-i}) \neq \phi_i (\theta_i', \boldsymbol\theta_{-i})$ follows. Thus, $P$ does not produce an individual contextual privacy violation at $(i, \btheta)$. 
\end{proof}

Given this relationship of individual contextual privacy and non-bossiness, we are able to leverage existing results to yield unique characterizations for the first-price auction and the serial dictatorship. 

Recall that in an \emph{object assignment setting} there is a finite set $C$ of objects, outcomes are $X = \{ x \colon N \to C \text{ injective}\}$ and types are $\theta_i \in \R^{\lvert C\rvert}$, for some set of \emph{objects} $C$. We assume strictness, $\theta_c \neq \theta_{c'}$ for $c \neq c' \in C$ with utilities $u_i(x ; \theta_i) = (\theta_i)_{x(i)}$. We say an object assignment choice rule is \emph{neutral} if for all type profiles $\btheta \in \bTheta,$  and permutations $\pi \colon N \to N$, $\phi(\pi\cdot\btheta) = \pi (\phi(\btheta))$.  
\begin{proposition}\label{prop:sd_icp}
An object assignment choice rule $\phi$ is individually contextually private, neutral and strategyproof if and only if it is a serial dictatorship.
\end{proposition}
\begin{proof}
The serial dictatorship is contextually private and non-bossy, hence individually contextually private by \Cref{lem:icp_nonbossy}. It also is strategyproof.

Conversely, if $\phi$ is individually contextually private, then it is also non-bossy by \Cref{lem:icp_nonbossy}. By \citet{svensson1999strategy}, a mechanism is neutral, strategyproof and non-bossy only if it is a serial dictatorship. So, if $\phi$ is individually contextually private, neutral and strategyproof, then it is a serial dictatorship.
\end{proof}
Note here that efficiency is implied by neutrality \citep{bade2016gibbard}. A characterization of efficient, strategyproof and non-bossy mechanisms is supplied by \citet{pycia2017incentive} and \citet{bade2020random}---note that the class characterized in those works is larger than just serial dictatorships. 

Similarly, we can draw on existing literature on non-bossiness to uniquely characterize the first-price auction. A single-item \emph{auction choice rule} is some $\phi: \bTheta \to N \times \R_+$ for $\Theta_i \subseteq \R_+$, $i =1, 2, \dots, n$, with outcomes $(i, t) \in  N \times \mathbb R_{+}$, which identify a winner $i$ and a price $t$. An auction rule is \textit{efficient} if the winner has the highest valuation. 
\begin{proposition}
The first-price auction is the unique individually contextually private, efficient, and individually rational auction rule. 
\end{proposition}
\begin{proof} 
By \Cref{prop:fbacp}, the first-price auction is contextually private. It is also non-bossy. Hence, by \Cref{lem:icp_nonbossy}, it is also individually contextually private. It is well known that the first-price auction is efficient and individually rational. 

Let $\phi$ be any individually contextually private, efficient and individually rational auction rule. By the characterization of efficient, non-bossy and individually rational auctions from \citet[Theorem 1]{pycia2022non}, this means that it is a protocol for the first-price auction.
\end{proof}

\subsection{Group Contextual Privacy}\label{subsec:groupcp}
Another variation is the \emph{group} contextual privacy violation, already introduced in \Cref{sec:whycp}. Group contextual privacy requires that if a protocol distinguishes two type \emph{profiles}, they must lead to different outcomes. 

\begin{definition}[Group Contextual Privacy Violations]
A protocol $P=(V,E)$ for a social choice function $\phi$ with strategies $\bsigma$ produces a \emph{group contextual privacy violation} at $\btheta \in \bTheta$ if there is a type $\btheta' \in \bTheta$ such that 
\[
\text{$P$ distinguishes $\btheta$ and $\btheta'$ yet $\phi(\btheta)=\phi(\btheta')$.}
\]
\end{definition}
This definition strengthens contextual privacy. Here, it is because it strengthens the underlying notion of distinguishability---two type profiles $\boldsymbol\theta, \boldsymbol\theta'$ are distinguishable if they belong to different terminal nodes. Regular contextual privacy's notion of distinguishability is on the agent-level---two types $\theta_i, \theta_i'$ are distinguishable if they belong to different terminal nodes, \emph{holding all other agent's types fixed} at $\btheta_{-i}.$ 

We characterize the set of group contextually private protocols next. 
\begin{proposition}\label{prop:groupcp}
An iterative partition protocol $P$ is group contextually private if and only if for every node $\tilde \bTheta$,
\[
\bigcupdot_{\tilde \bTheta' \in \children(\tilde \bTheta)} \phi(\tilde \Theta')  = \phi(\tilde \bTheta)
\]
is a disjoint union.
\end{proposition}
\begin{proof}
First assume that $P$ is group contextually private, and assume for contradiction that $\tilde \bTheta$ has distinct children $\tilde \bTheta', \tilde \bTheta''$ such that $\phi(\tilde \bTheta') \cap \phi(\tilde \bTheta'') \neq \emptyset$. Hence, there are $\boldsymbol\theta' \in \tilde \bTheta' $ and $\boldsymbol\theta'' \in \tilde \bTheta'' $ such that $\phi(\btheta') = \phi(\btheta'')$, which contradicts group contextual privacy.

Next assume that reachable outcomes are disjoint at each node. Let $\boldsymbol\theta'$ and $\boldsymbol\theta''$ be distinguished at $\tilde \bTheta$. As outcomes are disjoint, it must be that $\phi(\boldsymbol\theta') \neq \phi(\boldsymbol\theta'')$. Hence the protocol is group contextually private.
\end{proof}
This result says that a choice rule $\phi$ is group contextually private if and only if it can be represented by a protocol in which, at every node, the agent's choice rules out a subset of the outcomes.

This characterization implies that the serial dictatorship and the first price auction choice rule are group contextually private. In the serial dictatorship protocol, whenever an agent is called to play, they obtain their favorite object among those that remain. So, their choice rules out the outcomes in which a different agent gets their favorite object that remains. In the first price auction, the agent agent renders a particular outcome impossible, namely the outcome under which they win the good at a particular price. 

\section{Conclusion}\label{sec:conclusion}
This article starts from a simple minimal-revelation principle: avoid disclosing superfluous information where possible. From this principle, and the core idea of a ``contextual privacy violation" that allows us to express it, our analysis begins to build a formal language for describing the tradeoffs between privacy and economic objectives, deriving new protocols that are maximally contextually private subject to implementation. The analysis is far from complete---rather, it lays out many paths for future work exploring this simple minimal-disclosure principle in mechanism design and computation.

A first avenue for future study concerns the mediating technologies that are available to the designer. Throughout, we assume that every report by an agent is fully observed by the designer. However, in many applications, trusted technologies mediate communication so that what an agent submits need not be fully revealed to the designer. Trusted third parties or advanced cryptographic protocols are extreme examples of this---they are mediating technologies which can in some cases reveal to the designer nothing but what's necessary---rendering trivial the analysis conducted here. But there are intermediate communication environments worth studying too; for instance, with a ``ballot box'' the designer can observe only aggregates—e.g., that some number of agents have value at least $x$—without learning \emph{which} agents those are.\footnote{A working paper version of this article analyzes contextual privacy when the designer can construct what we call ``count" protocols with such ballot boxes \citep[Appendix D]{wp_cp}.} A thorough understanding of how contextual privacy-maximizing protocols change with assumptions on the mediating technology would be valuable. 

A second direction for further research considers privacy relative to different observers. Our analysis treats the designer as the sole privacy-relevant observer, but in practice, privacy from other agents (and coalitions of agents) and from third parties is often equally—or more—salient. A unified treatment that considers contextual privacy by the observer---designer, other players, and outsiders---would make explicit the trade-offs across these viewpoints. 

Third, our analyses in this paper are largely independent of the statistical structure of the environment. In this sense it is extreme: any time the designer learns information, it must make some difference. This makes contextual privacy as written here overly sensitive to non-deterministic mechanisms---if, with some small probability $p$, the outcome depends on the entire type profile, any protocol is fully contextually private. A smoother notion of contextual privacy could make explicit use of the designer’s prior. One natural approach may be an $\varepsilon$-$\delta$ formulation: information disclosure of size at most $\varepsilon$ about the type profile is deemed justified only if it shifts the distribution over outcomes by at least $\delta$ (in expectation under the prior). This would better capture settings where learning and data externalities are first-order, and it would permit a systematic treatment of \emph{approximate} implementation, trading exact execution of the choice rule for measured privacy loss. The resulting statistical-contextual privacy-optimal mechanisms may have a ``learning" phase and an ``exploitation" phase---hinting at a link between the ascending–join protocol and learning algorithms. This approach would also offer a more precise comparison between contextual privacy and differential privacy and its variants.

A final direction pushes toward a systematic treatment of privacy and incentives, which we have only partially addressed here. While our positive constructions incidentally enjoy good incentive properties, we do not yet have a principle that selects, among feasible mechanisms, those that are maximally contextually private subject to dominant‐strategy or Bayesian incentive constraints. We also do not have an analysis of when, generally, incentive constraints are complementary to or in tension with the demands of contextual privacy.

\pagebreak 

\bibliographystyle{aer}
\bibliography{refs_aer.bib}

@article{satterthwaite1981strategy,
  title={Strategy-proof allocation mechanisms at differentiable points},
  author={Satterthwaite, Mark A and Sonnenschein, Hugo},
  journal={The Review of Economic Studies},
  volume={48},
  number={4},
  pages={587--597},
  year={1981},
  publisher={JSTOR}
}

@article{gretschko2014information,
  title={Information acquisition during a descending auction},
  author={Gretschko, Vitali and Wambach, Achim},
  journal={Economic Theory},
  volume={55},
  number={3},
  pages={731--751},
  year={2014},
  publisher={Springer}
}

@article{cramton1998ascending,
  title={Ascending auctions},
  author={Cramton, Peter},
  journal={European Economic Review},
  volume={42},
  number={3-5},
  pages={745--756},
  year={1998},
  publisher={Elsevier}
}

@inproceedings{kleinberg2016descending,
  title     = {Descending Price Optimally Coordinates Search},
  author    = {Kleinberg, Robert and Waggoner, Bo and Weyl, E. Glen},
  booktitle = {Proceedings of the 2016 ACM Conference on Economics and Computation (EC '16)},
  pages     = {23--24},
  year      = {2016},
  publisher = {Association for Computing Machinery},
  doi       = {10.1145/2940716.2940760},
  url       = {https://doi.org/10.1145/2940716.2940760}
}

@techreport{pycia2022non,
  author      ={Pycia, Marek and Raghavan, Madhav},
  title       ={Non-bossiness and first-price auctions},
  institution ={CEPR Discussion Paper No. DP16874},
  year        ={2022}
}

@article{fluvia2012buyer,
  title={Buyer and seller behavior in fish markets organized as Dutch auctions: Evidence from a wholesale fish market in Southern Europe},
  author={Fluvi{\`a}, Modest and Garriga, Anna and Rigall-I-Torrent, Ricard and Rodr{\'\i}guez-Car{\'a}mbula, Ernesto and Sal{\'o}, Albert},
  journal={Fisheries Research},
  volume={127},
  pages={18--25},
  year={2012},
  publisher={Elsevier}
}

@unpublished{pycia2024ordinal,
  title        = {Ordinal Simplicity and Auditability in Discrete Mechanism Design},
  author       = {Pycia, Marek and Ünver, M. Utku},
  note         = {Available at SSRN: \url{https://ssrn.com/abstract=4609262}},
  year         = {2023},
  month        = {October},
  doi          = {10.2139/ssrn.4609262}
}

@article{rothkopf1990vickrey,
  title={Why are Vickrey auctions rare?},
  author={Rothkopf, Michael H and Teisberg, Thomas J and Kahn, Edward P},
  journal={Journal of Political Economy},
  volume={98},
  number={1},
  pages={94--109},
  year={1990},
  publisher={The University of Chicago Press}
}

@article{ausubel2004efficient,
  title={An efficient ascending-bid auction for multiple objects},
  author={Ausubel, Lawrence M},
  journal={American Economic Review},
  volume={94},
  number={5},
  pages={1452--1475},
  year={2004}
}

@article{alvarez2020comprehensive,
  title={Comprehensive survey on privacy-preserving protocols for sealed-bid auctions},
  author={Alvarez, Ramiro and Nojoumian, Mehrdad},
  journal={Computers \& Security},
  volume={88},
  pages={101502},
  year={2020},
  publisher={Elsevier}
}

@article{grigoryan2023theory,
  title={A theory of auditability for allocation and social choice mechanisms},
  author={Grigoryan, Aram and M{\"o}ller, Markus},
  journal={arXiv preprint arXiv:2305.09314},
  year={2023}
}

@inproceedings{izmalkov2005rational,
  title={Rational secure computation and ideal mechanism design},
  author={Izmalkov, Sergei and Micali, Silvio and Lepinski, Matt},
  booktitle={46th Annual IEEE Symposium on Foundations of Computer Science (FOCS'05)},
  pages={585--594},
  year={2005},
  organization={IEEE}
}

@article{izmalkov2011perfect,
  title={Perfect implementation},
  author={Izmalkov, Sergei and Lepinski, Matt and Micali, Silvio},
  journal={Games and Economic Behavior},
  volume={71},
  number={1},
  pages={121--140},
  year={2011},
  publisher={Elsevier}
}

@article{segal2007communication,
  title={The communication requirements of social choice rules and supporting budget sets},
  author={Segal, Ilya},
  journal={Journal of Economic Theory},
  volume={136},
  number={1},
  pages={341--378},
  year={2007},
  publisher={Elsevier}
}

@article{mackenzie2022menu,
  title={Menu mechanisms},
  author={Mackenzie, Andrew and Zhou, Yu},
  journal={Journal of Economic Theory},
  volume={204},
  pages={105511},
  year={2022},
  publisher={Elsevier}
}

@article{bade2020random,
  title={Random serial dictatorship: {T}he one and only},
  author={Bade, Sophie},
  journal={Mathematics of Operations Research},
  volume={45},
  number={1},
  pages={353--368},
  year={2020},
  publisher={INFORMS}
}

@article{svensson1999strategy,
  title={Strategy-proof allocation of indivisible goods},
  author={Svensson, Lars-Gunnar},
  journal={Social Choice and Welfare},
  volume={16},
  pages={557--567},
  year={1999},
  publisher={Springer}
}

@article{overdescending,
    author = {Harstad, Ronald},
    title = {Systems and methods for advanced auction management},
    journal = {US Patent No. 10,726,476 B2},
    year = {2018}
}

@article{balinski1999tale,
  title={A tale of two mechanisms: student placement},
  author={Balinski, Michel and S{\"o}nmez, Tayfun},
  journal={Journal of Economic theory},
  volume={84},
  number={1},
  pages={73--94},
  year={1999},
  publisher={Elsevier}
}

@article{pycia2017incentive,
  title={Incentive compatible allocation and exchange of discrete resources},
  author={Pycia, Marek and {\"U}nver, M Utku},
  journal={Theoretical Economics},
  volume={12},
  number={1},
  pages={287--329},
  year={2017},
  publisher={Wiley Online Library}
}

@article{shapley1974cores,
  title={On cores and indivisibility},
  author={Shapley, Lloyd and Scarf, Herbert},
  journal={Journal of Mathematical Economics},
  volume={1},
  number={1},
  pages={23--37},
  year={1974},
  publisher={Elsevier}
}

@article{genmedian,
 ISSN = {00485829, 15737101},
 URL = {http://www.jstor.org/stable/30023824},
 author = {Herv\'{e} Moulin},
 journal = {Public Choice},
 number = {4},
 pages = {437--455},
 publisher = {Springer},
 title = {On strategy-proofness and single peakedness},
 urldate = {2023-10-18},
 volume = {35},
 year = {1980}
}

@inproceedings{chor1989zero,
author = {Chor, Benny and Kushilevitz, E.},
title = {A Zero-One Law for Boolean Privacy},
year = {1989},
isbn = {0897913078},
publisher = {Association for Computing Machinery},
address = {New York, NY, USA},
url = {https://doi.org/10.1145/73007.73013},
doi = {10.1145/73007.73013},
booktitle = {Proceedings of the Twenty-First Annual ACM Symposium on Theory of Computing},
pages = {62–72},
numpages = {11},
location = {Seattle, Washington, USA},
series = {STOC '89}
}

@inproceedings{brandt2005unconditional,
  title={Unconditional privacy in social choice},
  author={Brandt, Felix and Sandholm, Tuomas},
  booktitle={Theorectical Aspects of Rationality and Knowledge},
  pages={207--218},
  year={2005}
}

@article{mcmillan1994selling,
  title={Selling spectrum rights},
  author={McMillan, John},
  journal={Journal of Economic Perspectives},
  volume={8},
  number={3},
  pages={145--162},
  year={1994},
  publisher={American Economic Association}
}

@article{milgrom2020clock,
  title={Clock auctions and radio spectrum reallocation},
  author={Milgrom, Paul and Segal, Ilya},
  journal={Journal of Political Economy},
  volume={128},
  number={1},
  pages={1--31},
  year={2020},
  publisher={The University of Chicago Press Chicago, IL}
}

@article{brandt2008existence,
  title={On the existence of unconditionally privacy-preserving auction protocols},
  author={Brandt, Felix and Sandholm, Tuomas},
  journal={ACM Transactions on Information and System Security (TISSEC)},
  volume={11},
  number={2},
  pages={1--21},
  year={2008},
  publisher={ACM New York, NY, USA}
}

@article{bade2016gibbard,
  title={Gibbard-Satterthwaite success stories and obvious strategyproofness},
  author={Bade, Sophie and Gonczarowski, Yannai A},
  journal={arXiv preprint arXiv:1610.04873},
  year={2016}
}

@article{mackenzie2020revelation,
  title={A revelation principle for obviously strategy-proof implementation},
  author={Mackenzie, Andrew},
  journal={Games and Economic Behavior},
  volume={124},
  pages={512--533},
  year={2020},
  publisher={Elsevier}
}

@article{liu2020preserving,
  title={Preserving bidder privacy in assignment auctions: design and measurement},
  author={Liu, De and Bagh, Adib},
  journal={Management Science},
  volume={66},
  number={7},
  pages={3162--3182},
  year={2020},
  publisher={INFORMS}
}

@techreport{woodward2020self,
  title={Self-auditable auctions},
  author={Woodward, Kyle},
  year={2020},
  institution={Working paper}
}

@article{li2017obviously,
  title={Obviously strategy-proof mechanisms},
  author={Li, Shengwu},
  journal={American Economic Review},
  volume={107},
  number={11},
  pages={3257--87},
  year={2017}
}

@article{akbarpour2020credible,
  title={Credible auctions: A trilemma},
  author={Akbarpour, Mohammad and Li, Shengwu},
  journal={Econometrica},
  volume={88},
  number={2},
  pages={425--467},
  year={2020},
  publisher={Wiley Online Library}
}

@article{nissenbaum2004privacy,
  title={Privacy as contextual integrity},
  author={Nissenbaum, Helen},
  journal={Washington Law Review},
  volume={79},
  pages={119},
  year={2004},
  publisher={HeinOnline}
}

@inproceedings{dwork2006differential,
  title={Differential privacy},
  author={Dwork, Cynthia},
  booktitle={International Colloquium on Automata, Languages, and Programming},
  pages={1--12},
  year={2006},
  organization={Springer}
}

@article{HeSandomirskiyTamuz,
  author  = {Kevin He and Fedor Sandomirskiy and Omer Tamuz},
  title   = {Private Private Information},
  journal = {Journal of Political Economy},
  year    = {2025},
  note    = {Forthcoming},
  url     = {https://arxiv.org/abs/2112.14356}
}

@article{StrackYang2024_Ecta_PrivacySignals,
  author  = {Philipp Strack and Kai Hao Yang},
  title   = {Privacy-Preserving Signals},
  journal = {Econometrica},
  year    = {2024},
  volume  = {92},
  number  = {6},
  pages   = {1907--1938},
  month   = nov,
  doi     = {10.3982/ECTA22017},
  url     = {https://doi.org/10.3982/ECTA22017}
}

@article{wp_cp,
    author = {Haupt, Andreas and Hitzig, Zo\"{e}},
    title = {Contextually Private Mechanisms (Working Paper)},
    year = {2024},
    journal = {arxiv},
    url = {https://arxiv.org/pdf/2112.10812v7}
}

@article{pai2013privacy,
  title={Privacy and mechanism design},
  author={Pai, Mallesh M and Roth, Aaron},
  journal={ACM SIGecom Exchanges},
  volume={12},
  number={1},
  pages={8--29},
  year={2013},
  publisher={ACM New York, NY, USA}
}

@article{eilat2021bayesian,
  title={Bayesian privacy},
  author={Eilat, Ran and Eliaz, Kfir and Mu, Xiaosheng},
  journal={Theoretical Economics},
  volume={16},
  number={4},
  pages={1557--1603},
  year={2021},
  publisher={Wiley Online Library}
}

@inproceedings{bogetoft2009secure,
  title={Secure multiparty computation goes live},
  author={Bogetoft, Peter and Christensen, Dan Lund and Damg{\aa}rd, Ivan and Geisler, Martin and Jakobsen, Thomas and Kr{\o}igaard, Mikkel and Nielsen, Janus Dam and Nielsen, Jesper Buus and Nielsen, Kurt and Pagter, Jakob and others},
  booktitle={International Conference on Financial Cryptography and Data Security},
  pages={325--343},
  year={2009},
  organization={Springer}
}

@article{segal2010communication,
  volume = {5},
  ISSN = {1933-6837},
  title = {Nash implementation with little communication},
  author = {Segal, Ilya R.},
  url = {http://dx.doi.org/10.3982/TE576},
  doi = {10.3982/te576},
  number = {1},
  journal = {Theoretical Economics},
  publisher = {The Econometric Society},
  year = {2010},
  pages = {51–71}
}

@misc{texasvgoogle2022,
  title        = {State of Texas et al. v. Google LLC},
  year         = {2022},
  howpublished = {United States District Court for the Southern District of New York},
  note         = {No. 1:21-cv-06841, decided September 13, 2022},
  url          = {https://law.justia.com/cases/federal/district-courts/new-york/nysdce/1:2021cv06841/565005/209/}
}

@article{gonczarowski2019,
title = {A stable marriage requires communication},
journal = {Games and Economic Behavior},
volume = {118},
pages = {626-647},
year = {2019},
issn = {0899-8256},
doi = {https://doi.org/10.1016/j.geb.2018.10.013},
url = {https://www.sciencedirect.com/science/article/pii/S0899825619301034},
author = {Yannai A. Gonczarowski and Noam Nisan and Rafail Ostrovsky and Will Rosenbaum}
}

@book{Learmount1985history,
  title     = {A History of the Auction},
  author    = {Brian Learmount},
  year      = {1985},
  publisher = {Barnard \& Learmount},
  address   = {Iver, UK},
  isbn      = {0951024000},
  language  = {en}
}

\pagebreak 

\appendix
\startcontents

\section*{Appendices}
\printcontents{}{0}{}

\section{Proofs}\label{sec:proofs}

\subsection{Impossibility of Contextually Private Stable Assignment}
\stabsecp*
\begin{proof}
The proof constructs collective but not individual pivotality. 

Let $s_1> s_2>s_3>s_4$. Fix the type profile of all $n-2$ agents that are not $i$ or $j$ to be $\btheta_{-ij}$ where each agent has a score greater than $s_1$ for their top choice object, and their top choice object has capacity to accommodate them. Assume further that the remaining two spots are for different objects. Label these objects with remaining spots $a$ and $b$.

Consider the final two agents $i, j \in N$. Choose their type profiles to be:
\begin{align*}
\theta_i &= (a \succ_i b, s_i(a)=s_1, s_j(b)=s_4), &
\theta_i'&= (b \succ_i a, s_i(a)=s_3, s_j(b)=s_2)\\
\theta_j &= (b \succ_j a, s_j(a)=s_4, s_j(b)=s_1), &
\theta_j'&= (a \succ_j b, s_j(a)=s_2, s_j(b)=s_3).
\end{align*}
The preferences of agents $i$ and $j$ for objects $a$ and $b$ are represented in \Cref{fig:stability}. Agent $i$ and $j$'s scores and preferences for other schools are arbitrary. 

Let $x$ be the outcome in which agent $i$ is matched to school $a$ and $j$ is matched to $b$. Let $y$ be the outcome in which agent $i$ is matched to $b$ and $j$ is matched to $a$. In both $x$ and $y,$ all agents not $i$ or $j$ are assigned to their top choice object at which they have a high score.
\begin{figure}
    \centering
    \includegraphics[scale=.2]{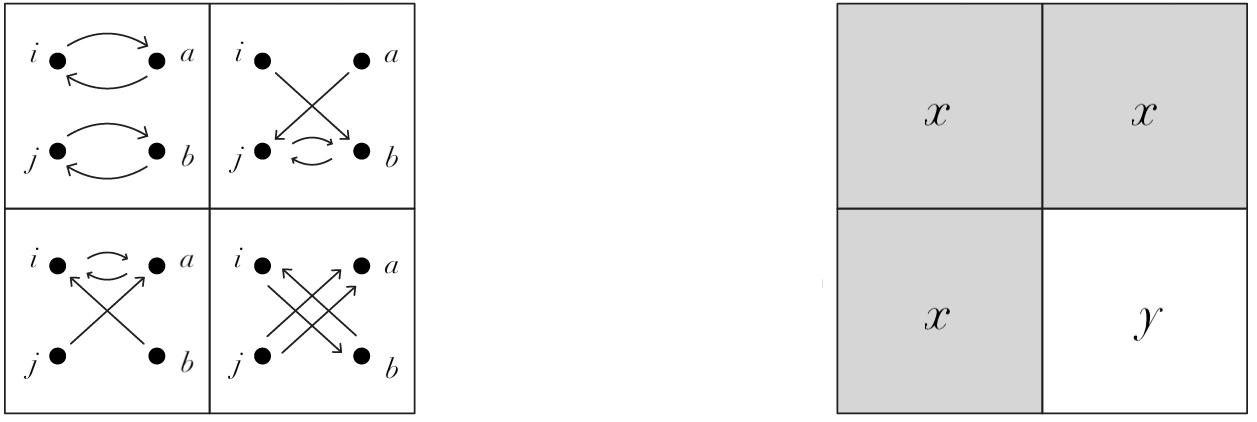}
    \caption{Constructing collective but not individual pivotality for college assignment. Agent types $\theta_i, \theta_i', \theta_j, \theta_j'$ (left, arrows from agents denote favored object, arrows from objects denote high score); outcomes under any stable choice rule (right, where $x = ((i,a), (j,b))$ and $y= ((j,a), (i,b))$).}
    \label{fig:stability}
\end{figure}

Stability requires that $\phi(\theta_i, \theta_j, \boldsymbol\theta_{-ij}) = (\theta_i, \theta_j', \boldsymbol\theta_{-ij})= (\theta_i', \theta_j, \btheta_{-ij})=x$ while $\phi(\theta_i', \theta_j', \boldsymbol\theta_{-ij}) =y.$ We have chosen a particular $\boldsymbol\theta_{-ij} \in \bTheta_{-ij},$ and particular $\theta_i, \theta_i', \theta_j, \theta_j' \in \Theta$ such that the condition of \Cref{thm:characterization} holds.
\end{proof}

\subsection{Sufficiency of Bimonotonic Protocols for Interval Pivotal Choice Rules}

\bimon*

\begin{proof}
We begin with an arbitrary protocol and transform it into a bimonotonic protocol with the same set of privacy violations by using a process of \emph{filling-in}.

Filling-In consists of repeated injection of nodes. For notational simplicity, we will denote a node of an iterative partition by $v$ (instead of $\tilde \bTheta$).

\begin{definition}[$v$-before-$v'$ injected protocol]
    Let $P$ be an iterative partition, let $v$ be a partition of $\Theta_i$, and $v'$ be a non-root node in $P$. Define the $v$-before-$v'$-injected protocol $P_{v, v'}$ to be the protocol in which the parent of $v'$ is partitioned first by $v$ and then by $v'$. See \Cref{fig:injection} for an illustration of injection. 
\end{definition}
\begin{figure}[h!]
    \centering
    \includegraphics[width=.5\linewidth]{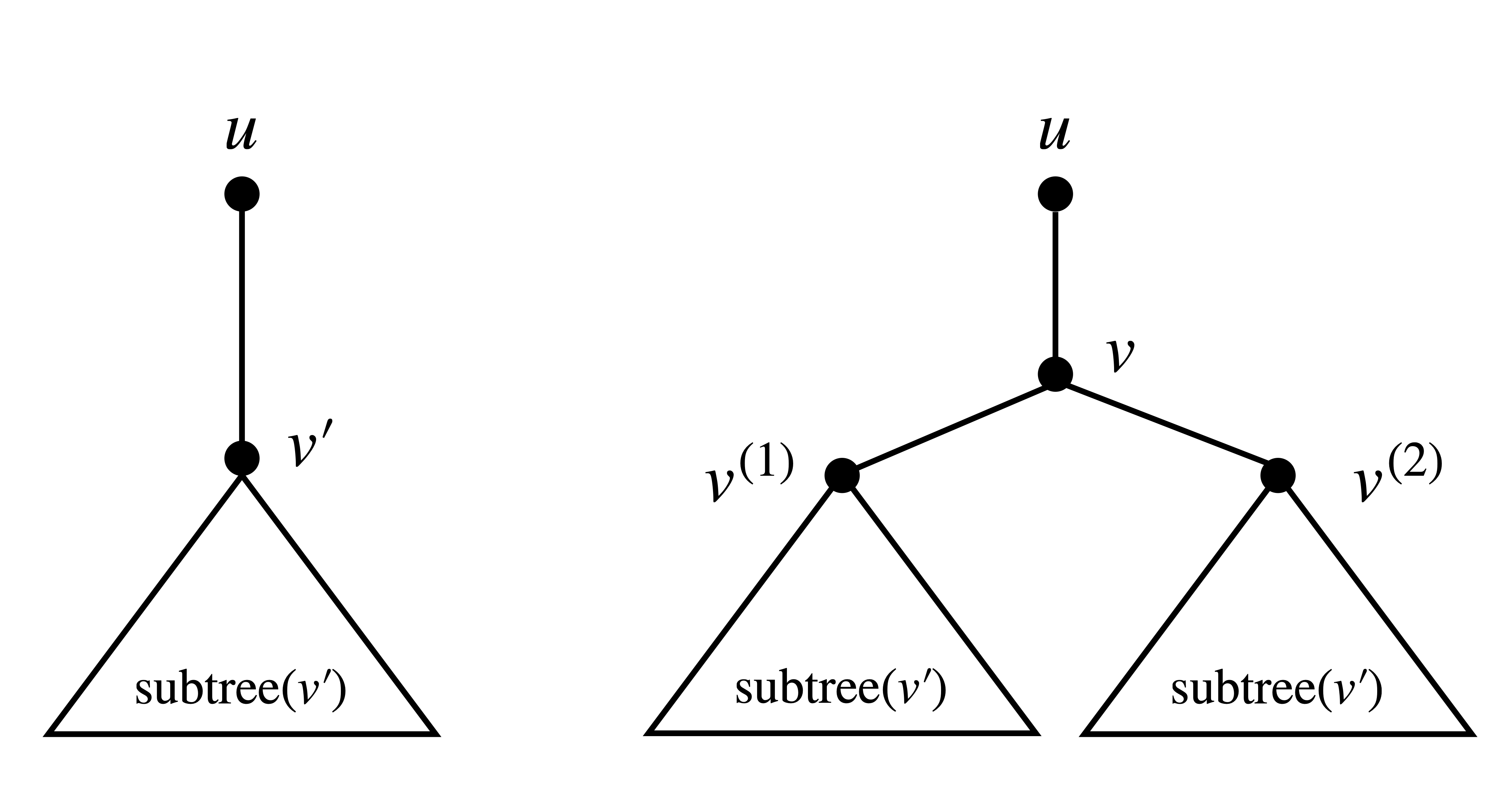}
    \caption{Protocol $P$ (left); $v$-before-$v'$ injected protocol $P_{v,v'}$ (right). }
    \label{fig:injection}
\end{figure}
\begin{definition}
    Let $\Theta$ be a finite ordered type space and denote by $\tsucc (\theta)$ the adjacent type higher than $\theta$ in $\Theta$ (if it exists) and by $\pred (\theta)$ the adjacent type lower than $\theta$ in $\Theta$ (if it exists). We call $\underline{\theta}_{v,i}$ the \emph{lowest separator} if 
    \begin{equation*}
    \underline{\theta}_{v,i} = \min \{ \theta_i \in \Theta_i   \mid  \text{ $v$ distinguishes } \theta_i \text{ and }\tsucc( \theta_i )\}  
    \end{equation*} 
   and we call $\overline{\theta}_{v,i}$ the \emph{highest separator} if 
   \begin{equation*}
    \overline{\theta}_{v,i} = \max \{ \theta_i \in \Theta_i  \mid \text{ $v$ distinguishes } \pred (\theta_i) \text{ and } \theta_i\}. 
    \end{equation*}
\end{definition}
That is, the lowest separator at node $v$ is the lowest element in the type space that is distinguished from its next highest type at $v$. The highest separator at node $v$ is the highest element in the type space that is distinguished from its next lowest type at node $v$. If these sets are empty, the query is trivial in that it does not affect the knowledge of the designer. It is then without loss to drop this query from the protocol.  

Define the quantities of the \emph{pairwise} lowest resp. highest separators,
\begin{equation}
\begin{split}
    \underline{\theta}_{v, v'} & = \min \{  \underline{\theta}_{v,i},  \underline{\theta}_{v',i}\} \\
    \overline{\theta}_{v, v'} & = \max \{  \overline{\theta}_{v,i},  \overline{\theta}_{v',i}\}. \label{eq:limitdefhigh}
    \end{split}
\end{equation}
We have the following result:
\begin{lemma}
\label{lem:injection}
    Let $\phi$ be interval pivotal, $i(v) = i(v')$, and $v' \succ v$. Then, for any $\tilde \theta \in [\underline{\theta}_{v, v'}, \overline{\theta}_{v, v'}]$, a for a partition $v'' = \{\{\theta_i \in \Theta_i : \theta_i \le \tilde\theta\}, \{\theta_i \in \Theta_i : \theta_i > \tilde\theta\}\} $
    \[
   \Gamma(P, \bsigma, \phi) = \Gamma(P_{v'', v'}, \bsigma', \phi).
    \]
\end{lemma}
In other words, consider an agent that is asked two queries $v$ and $v'$ one after the other on one path through the partition. Inserting a threshold query \emph{in between} $v$ and $v'$, where the threshold $\tilde \theta$ lies \emph{between} the highest and lowest types that $v$ and $v'$ can distinguish, does not change to the set of contextual privacy violations.

\begin{proof}
    We show first that $\Gamma(P_{v'',v'}, \bsigma', \phi) \subseteq \Gamma(P, \bsigma, \phi)$. Let $(\btheta, i) \in \Gamma(P_{v'',v'}, \bsigma', \phi) \setminus \Gamma(P, \bsigma, \phi)$ for contradiction. Then there is a $\theta_i'$ such that $(\theta_i, \btheta_{-i})$ and $(\theta_i', \btheta_{-i})$ are distinguished at  $v''$ and  
    \[
    \phi (\theta_i, \btheta_{-i}) = \phi (\theta_i', \btheta_{-i}).
    \]
    By symmetry, it is without loss to assume $\theta_i \le  \theta_i'$. As $\theta_i$ and $\theta_i'$ are distinguished, and by the definition of threshold queries, it must be that $\theta_i \le \tilde \theta < \theta_i'$.

    We will now show that $\btheta$ has a contextual privacy violation for agent $i$ at node $v$ or node $v'$ in $P$, leading to a contradiction. As $\phi (\theta_i, \btheta_{-i}) = \phi (\theta_i', \btheta_{-i})$ and because of interval pivotality, there are two cases, corresponding to the \enquote{upper} or \enquote{lower} interval over which $\theta_i \mapsto \phi(\theta_i, \btheta_{-i})$ is constant:
    \begin{itemize}
        \item[(a)] $\phi(\hat \theta_i, \btheta_{-i}) = \phi(\theta_i', \btheta_{-i})$ for all types $\hat \theta_i \le \theta_i'$ (in particular this holds for $\hat{\theta}_i = \underline{\theta}_{v, v'}$ by the definition of the lowest separator), or
        \item[(b)] $\phi(\hat \theta_i, \btheta_{-i}) = \phi( \theta_i, \btheta_{-i})$ for all types $\hat \theta_i > \theta_i$ (in particular this holds for $\hat{\theta}_i = \overline{\theta}_{v, v'}$ by the definition of the highest separator).
    \end{itemize}
    For the first case (a), for notational simplicity call $v$ the node that attains the minimum in \eqref{eq:limitdefhigh}. By definition of $\underline \theta_{v, v'}$, the types $ \underline{\theta}_v, \tsucc (\underline\theta_v)$ are distinguished at $v$. There are two further cases within case (a):
    \begin{itemize}
        \item $(\theta_i, \btheta_{-i})$ and $(\theta_i', \btheta_{-i})$ are distinguished at $v$. In this case, both type profiles produce contextual privacy violations for agent $i$ at $v$.
        \item $(\theta_i, \btheta_{-i})$ and $(\theta_i', \btheta_{-i})$  are \emph{not} distinguished at $v$, but this means that $\theta_i$ is distinguished from either of $ \underline{\theta}_v, \tsucc (\underline\theta_v)$, or both. So, $(\btheta, i)$ produces a contextual privacy violation at $v$ with either $(\underline\theta_{v}, \btheta_{-i})$ or $(\tsucc(\underline\theta_{v}), \btheta_{-i})$.
    \end{itemize}
 For the second case (b), we can follow similar reasoning, but with flipped inequality signs and $v$ replaced with $v''$. In this case, one proceeds by showing that a contextual privacy violation is produced at the node that attains the maximum in \eqref{eq:limitdefhigh}.
    
For the converse direction, observe that $P_{v'',v'}$ reveals weakly more information to the designer than $P$. Hence, also $\Gamma(P, \bsigma, \phi) \subseteq \Gamma(P_{v'',v'}, \bsigma', \phi).$
\end{proof}

Using \Cref{lem:injection}, we will \enquote{fill-in} all queries between the highest and lowest separator. This results in a \emph{filled-in} protocol.
\begin{definition}
    We call a protocol $P$ \emph{filled-in} if every query is a threshold query and the threshold for every query $v$ and next query $v'$ must have thresholds that are adjacent in the agent's type space.
\end{definition}
\begin{lemma}
\label{lem:filledin}
    Let $\phi$ be an interval pivotal choice rule. Then, for any protocol $(P, \bsigma)$, there is a filled-in protocol $(P', \bsigma)$ such that $\Gamma(P, \bsigma, \phi)= \Gamma(P', \bsigma', \phi).$
\end{lemma}
\begin{proof}
We prove this lemma in three steps: anchoring, inserting and deleting.
\begin{description}
    \item[Step 0: Grounding.] Add trivial queries to all agents in the beginning, where $r_i$ can be thought of as a copy of the root node. These allow us to perform protocol injection on initial queries to agents and resolves issues for protocols where an agent only gets a single query.  Inserting trivial queries affects neither computability of the choice rule nor contextual privacy violations. 
    \item[Step 1: Anchoring.] We then introduce threshold queries at the highest ($\overline \theta_{v, v'}$) and lowest ($\underline \theta_{v,v'}$) separators for any pair of node $v$ and next node $v'$ for this agent, before $v'$. This leads to the introduction of at most $2\lvert V \rvert$ many new queries and does not affect whether $P$ is a protocol for $\phi$.
    \item[Step 2: Inserting.] For any pair of threshold queries, nodes $v$ and next node $v'$ to the same agent, whose thresholds are not adjacent in $\Theta,$ insert threshold queries for all thresholds between $\thresh(v)$ and $\thresh(v')$ (in the type space) before the later of $v,v'$. This leads to at most $2\lvert V\rvert^2\lvert \Theta\rvert$ many queries.
    \item[Step 3: Deleting.] For any non-threshold queries $v$, all thresholds between $\underline \theta_v$ and $\overline{\theta}_v$ were added in Steps 1 and 2. This implies that all non-threshold queries can be deleted without affecting computability of the choice rule. Additionally, their deletion cannot affect contextual privacy violations because any distinguished types under the queries are distinguished at at least one of the added threshold queries.
\end{description}
     The resulting protocol is filled-in and has the same set of privacy violations. 
\end{proof}

A filled-in protocol has almost a bimonotonic iterative partition. But, it may have sequences of threshold queries to a single agent that increase and then decrease (or vice versa) in the thresholds. This already yields a bimonotonic protocol, as all non-monotonic queries after the first query are trivial (that is, do not distinguish any type profiles). The associated iterative partition is, hence, bimonotonic.

    \begin{figure}[h!]
    \centering
    \includegraphics[width=.7\linewidth]{img/scrubbing.png}
    \caption{Trivial queries to agent $i$ in protocol $P$ after initial query with \enquote{higher-than} answer in light grey (left); Trivial queries to agent $i$ in protocol $P$ after initial query with \enquote{lower-than} answer in light grey (right). The $x$-axis represents the order of the queries in protocol $P$. The $y$-axis is the threshold of threshold queries in protocol $P$. Each query in the original protocol is a dot, queries deleted during scrubbing are gray dots.}
    \label{fig:scrubbing}
\end{figure}  
\end{proof}

\subsection{Maximal Contextual Privacy of the Ascending-Join Auction}
\maxcpkpa*

\begin{proof}
For the second statement, we observe that the ascending-join protocol is a personal clock auction, which are known to have equilibria in obviously dominant strategies, see \citet[Theorem 3]{li2017obviously}.

The first statement of \Cref{thm:maxcpkpa} follows from the following lemma.

\end{proof}
\begin{lemma}
\label{lemma:maxcpkpa}
    Let $P$ be a weak contextual privacy improvement of $P_{\ASCJ}$, $\Gamma(P, \phik) \subseteq \Gamma(P_{\ASCJ}, \phik)$. Then, $P$ is contextual privacy equivalent to $P_{\ASCJ}$, $\Gamma(P, \phik) = \Gamma(P_{\ASCJ}, \phik)$.
\end{lemma}
In the proof, we will use the notation $\bTheta_v$ for a node in the iterative partition.
\begin{proof}
Let $P$ be such a protocol. We prove this statement by induction over the nodes $\bTheta_v = \tilde \bTheta$ of $P$. We consider $\bTheta_v = \bTheta$ and $\bTheta_v \neq \bTheta$ separately. In particular, we will show that to for any $P$, being a privacy improvement over $P_{\ASCJ}$ means that the nodes that can be reached have a particular structure, which we call \emph{$\W$- and $\LQL$-protective information sets}.
\begin{definition}
    We call a set of type profiles $\bTheta_v \subseteq \bTheta$ a $\W$- and $\LQL$-\emph{protective} if there is a type $\tilde\theta \in \Theta$ such that:
    \begin{enumerate}
        \item There are exactly $k+1$ agents $i$ such that $\Theta_{v,i} = \{\theta_i \mid \theta_i > \tilde\theta\}$.
        \item For agents $i = 1, 2, \dots, l$ with $l$ being the highest index among agents who ever got a question (i.e. $l \coloneqq \max\{l \colon \Theta_{v,l} \neq\Theta\}),$ it must be that, either: 
        \begin{itemize}
            \item[(i)] $\Theta_{v,i} = \{ \theta_i \mid \theta_i > \tilde\theta \}$ \enquote{\emph{(active agent)}}, or  
            \item[(ii)] $\Theta_{v,i} \subseteq \{  \theta_i \mid \theta_i \le \tilde\theta \}$ \enquote{(\emph{dropped-out agent})}.
        \end{itemize}
    \end{enumerate}
    We call all other agents {\emph{idle}}.
\end{definition}
For a protective information state $\bTheta_v$, we call $\tilde\theta$ the \emph{running price}. Sets that are protective have $k+1$ agents about whom it is known only that their type is strictly above the running price (active agents).  Nothing is known about agents with indices \enquote{to the right} of the agent with the highest index among agents who have ever been asked a question (idle agents). All other agents have a type weakly below the running price (dropped-out agents). We illustrate this definition in \Cref{fig:knormal}. 

\begin{figure}
    \centering
    \includegraphics[width=.65\linewidth]{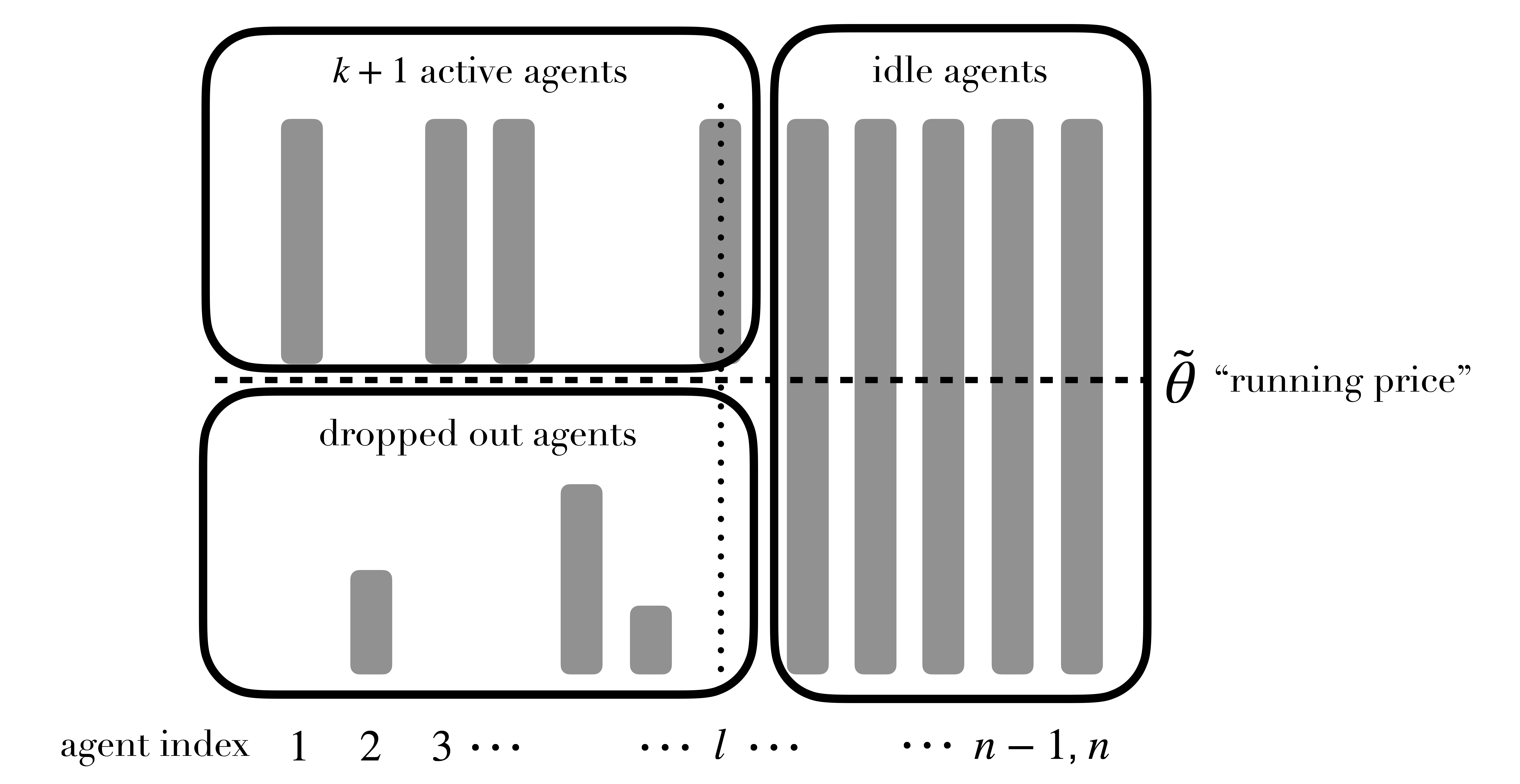}
    \caption{$\W$- and $\LQL$- Protective Information State}
    \label{fig:knormal}
\end{figure}

\begin{definition}
    We say that a query $v$ is $\phi$-redundant if for any $\theta_i, \theta_i'$, $\btheta_{-i}$ such that $(\theta_i, \btheta_{-i})$ and $(\theta_i', \btheta_{-i})$ are distinguished at $v$, we have
    \[
        \phi (\theta_i, \btheta_{-i}) = \phi (\theta_i', \btheta_{-i}).
    \]
\end{definition}
\end{proof}
We make the following observation, which is direct from the definitions.
\begin{lemma}
\label{lem:redundant}
    Let $P$ be a protocol for $\phik$ and $v$ be a $\phi$-redundant query in $P$. Any protocol $P'$ that removes $v$ and continues with \emph{any} of the children of $v$ from $P$ is still a protocol for $\phi$. $P'$ is weakly more contextually private than $P$.
\end{lemma}

We next make an additional observation about the structure of protection sets (that is, sets $A \subseteq \bTheta \times N \setminus \Gamma(P, \phi)$) to more easily reason about which queries are forbidden under a given protection set. To this end, the structure of a \emph{closed} protection set is helpful.

\begin{definition}[Closed protection sets]
    We call a protection set $A$ \emph{downward closed} if we have that for $((\theta_i, \btheta_{-i}), i) \in A$  and $\theta_i' \le \theta_i$,
    \[
     ((\theta_i', \btheta_{-i}), i) \in A.
    \]
    A protection set is $A$ \emph{upward closed} if we have that for $((\theta_i, \btheta_{-i}), i) \in A$ and $\theta_i' \ge \theta_i$,
    \[
     ((\theta_i', \btheta_{-i}), i) \in A.
    \]
    If a set is downward closed or upward closed, we call it \emph{closed}. 
\end{definition}
\begin{definition}[$\phi$-closed protection sets]
    If $A$ is a closed protection set, and for any $i =1, 2, \dots, n$, $\theta_i, \theta_i' \in \Theta$ and $\btheta_{-i} \in \bTheta_{-i}$, such that $((\theta_i, \btheta_{-i}), i), ((\theta'_i, \btheta_{-i}), i) \in A$, it holds 
    \[
    \phi(\theta_i, \btheta_{-i}) = \phi(\theta_i', \btheta_{-i}),
    \]
    we call it $\phi$-closed.
\end{definition}
The property of $\phi$-closedness may be seen as a {minimality} property for interval pivotal choice rules. Any distinction of type profiles in a $\phi$-closed protection set would lead to a contextual privacy violation for the queried agent.
\begin{lemma}
    The protection sets for $\W$ and $\LQL$ are $\phik$-closed.
\end{lemma}
\begin{proof}
    First, consider the set $\W = \{(\btheta,i): \btheta \in \bTheta, i \in \W(\btheta)\}$. This set is upward closed as a winner with a higher type will remain a winner, leaving the outcome unchanged. Similarly, $\LQL = \{(\btheta, \W(\btheta)): \btheta \in \bTheta, i \in \LQL(\btheta)\}$ is downward closed, as a late-queried loser with a lower type will also be a late-queried loser, leaving the outcome unchanged.
\end{proof}
\begin{lemma}
\label{lem:distinguishing}
    Let $A$ be a $\phi$-closed protection set. A protocol $P$ has a contextual privacy violation in $A$ if it contains a query $v$ and two type profiles $(\theta_i, \btheta_{-i}),(\theta_i', \btheta_{-i})$ that $P$ distinguishes such that $((\theta_i, \btheta_{-i}), i), ((\theta_i', \btheta_{-i}), i) \in A$.
\end{lemma}
\begin{proof}
    We show this lemma for upward closedness. The proof for downward closedness is analogous. We say that \enquote{$i$ is in $A$ under $\btheta$} if $(\btheta, i) \in A$. 
    
    Consider a query $v$ in protocol $P$ and two type profiles $\btheta, \btheta'$ that are distinguished at $v$ such that $i$ is in $A$ for both $(\theta_i, \btheta_{-i})$ and $ (\theta_i', \btheta_{-i})$. Without loss we may assume that $\theta_i < \theta_i'$. As $A$ is upward closed, it must be that $i$ is in $A$ also under $(\theta_i', \btheta_{-i})$, and does not change the outcome, by $\phi$-closedness. This means that there is a contextual privacy violation in the protection set $A$.
\end{proof}

\begin{lemma}
\label{lem:beginning}
    There is a unique sequence of non-redundant, $\W$- and $\LQL$-protecting queries on $\bTheta$ that leads to a $\W$- and $\LQL$-protective information state, or termination.
\end{lemma}
\begin{proof}
We may restrict to bimonotonic protocols by \Cref{thm:bimon}, in particular to protocols that only use threshold queries. 

By \Cref{lem:distinguishing}, any query with a threshold higher than the lowest type violates winner privacy because both a negative answer and an affirmative answer may lead to the queried agent winning. This continues to hold until $k$ higher-priority agents have given an affirmative answer to a threshold query about the lowest type. We will check that the $(k+1)^{\text{st}}$ query to an agent will still have the lowest type as a threshold, and leads to a $\W$- and $\LQL$-protective information state. For the first $k$ queries, we only need to determine the order in which agents are queried. 

Denote $h \in \{ 1, 2, \dots, n\}$ the number affirmative answers given so far. Call agents for which a type higher than the running price is not ruled out \enquote{remaining}. We show by induction over $h$ the following statement:
\begin{equation}\label{eq:k1hremaining}
\begin{split}
 &  \text{The only non-redundant query that does not violate winner or late-queried} \\
 &  \text{loser privacy is to the $(k + 1 - h)^{\text{th}}$ remaining agent in priority order.}
   \end{split}
\end{equation}
Note that all queries to agents who dropped out are redundant and it is without loss to abstract from them (by \Cref{lem:redundant}). Also note that all queries to agents $1, 2, \dots, k - h$ make it possible that the queried agent is a winner both in the case of an affirmative and a negative answer, hence violating $\W$-protection by \Cref{lem:distinguishing}. All queries to agents $k - h + 2, k- h + 3, \dots, n$ make it possible that the queried agent is a late-queried loser in the case of an affirmative or a negative answer, hence violating $\LQL$ privacy. Asking the $(k+1-h)^{\text{th}}$ remaining agent does not lead to a $\W$ or $\LQL$ contextual privacy violation.

(As a concrete example, consider the very first query. We may not ask agent $1, 2, \dots, k$ as they might win both with an affirmative and a negative answer. We may also not ask agents $k+2, k+3, \dots, n$, as they might be late-queried losers in case of a positive and a negative answer. Hence, agent $k+1$ must be queried. What's important is that a negative answer from the $(k+1)^{\text{st}}$ unambiguously means that they are not a winner. They cannot be a late-queried loser by definition.)

This determines the order of the first $k$ queries. Observe that agents $2, \dots, m$ for some $m \le n$ have been asked by the property \eqref{eq:k1hremaining}. If we show that agent 1 is the only agent that can be queried, and must be queried for the lowest type, then the resulting information of the designer is a $\W$- and $\LQL$-protective information state (unless the protocol can terminate before this is the case). 

We distinguish three cases. First, asking agents $1, \dots, m$ for a higher threshold might lead to them being a winner for both an affirmative and a negative answer rendering these queries forbidden by \Cref{lem:distinguishing}. Second, asking agents $m+1, \dots, n$ for any threshold might lead them to being a late-queried loser for both an affirmative or a negative answer, rendering these forbidden by \Cref{lem:distinguishing}. Third, asking agent 1 for any threshold higher than the lowest type might lead to them being a winner for both an affirmative and a negative answer, making this query forbidden, also by \Cref{lem:distinguishing}. Asking agent 1 a threshold query with the lowest type does not lead to $\W$ and $\LQL$ violations. Note that asking agent 1 complies with the agent order in \eqref{eq:k1hremaining}.

Either the designer reaches information state $\bTheta_v$ in which $k+1$ agents have given an affirmative answer. In this case, the designer's information state $\bTheta_v$ is $\W$- and $\LQL$-protective. Or, all agents except for $k$ have given negative answers. In this case, the protocol may terminate (allocating to the remaining $k$ agents at the current running price).
\end{proof}
We reduce the steps on internal nodes $\bTheta_v \neq \bTheta$ to \Cref{lem:beginning}.
\begin{lemma}
\label{lem:step}
    On a $\W$- and $\LQL$-protective set $\bTheta_v$, there is a unique sequence of non-redundant, $\W$- and $\LQL$-protecting queries, that leads to a $\W$- and $\LQL$-protective set or termination.
\end{lemma}
\begin{proof}
    We reduce this to \Cref{lem:beginning}. Observe that any query to a dropped-out agent is redundant. We may therefore drop them from consideration. Also observe that any query with a threshold lower than the running price is either trivial (if it is to an active agent), redundant (if it is to a dropped-out agent) or could lead to the queried agent being a late-queried winner for both an affirmative and a negative answer (if it is to an idle agent). Hence, we may reduce to a situation in which the designer faces only active agents and needs to compute the choice rule $\phik$ on a restricted type space $\Theta' \coloneqq \{ \theta' \mid \theta' > \tilde\theta\}$, where $\tilde\theta$ is the current running price, and apply \Cref{lem:beginning}.
\end{proof}

\begin{remark}\label{rem:ascj-contest-equivalence}
So far, we have not reasoned about whether the order of querying described in \eqref{eq:k1hremaining} corresponds to the ascending-join partition as presented in \Cref{sec:maxcp}. To see this, we first observe that the $(k+1)$st agent stating that they have a type at least the running price $t$ rules out all tentative assignments $(W, t)$, $W \subseteq \{1, 2, \dots, n\}$. So we need to show that for each running price, the query to agents given by \eqref{eq:k1hremaining} is the same as asking agents repeatedly in the order $i= 1, 2, \dots, n$ whether they can rule out a tentative outcome. We demonstrate this by induction. Assume that there have been $h\in \{ 0, 1, \dots, k\}$ affirmative answers so far, and $l \in \{ 0, 1, \dots, n-k-1\}$ negative answers to queries for the running price. We claim that the query according to \eqref{eq:k1hremaining} rules out the outcome allocating to the highest-priority remaining agents and that no agent of a lower index has sufficient information by themselves to do so. First observe that a query to an agent that is not active anymore cannot rule out the current active assignment (as their type is known and they are losers). Remaining agent $i = 1, 2, \dots, k - h$ are unable to rule out the outcome that agents $1, 2, \dots, k-h$ win (which we argued in the proof of \Cref{lem:beginning}) in addition to the $h$ agents who already gave affirmative answers. The first agent to rule this outcome is the $(k + 1 - h)$th remaining agent, who can rule out all assignments where they are not a winner at the current running price. The unique smallest (in the strong set order) set of winners that is not ruled out is that the highest $k$ of all agents who have not given a negative answer yet win.
\end{remark}

\section{Impossibility Results in Non-Auction Domains}\label{sec:full_cpimpossibilities}
Many social choice functions exhibit collective but not individual pivotality on some subset of the type space. We give some examples of such choice functions commonly studied in the literature and used in practice which, by \Cref{cor:collectindividual}, do not admit a fully contextually private protocol. That is, any protocol that computes these choice rules must produce contextual privacy violations. 

\subsection{Housing Assignment}\label{subsec:houseassignment}
Consider first the house assignment problem \cite{shapley1974cores}. All agents are initially endowed with an object from $C$. The type space is the universal domain of all preferences over $C$. Denote the initial assignment by an injective function $e \colon N \to C$, where $e(i) \in C$ refers to agent $i$'s initial endowment. For our result it will be irrelevant whether the endowments are private information or known to the designer. We call a choice rule $\phi$ \emph{individually rational} if for all $i \in N$ and all $\btheta \in \bTheta,$
\[
\phi_i(\boldsymbol\theta) \succeq_{i} e(i). 
\]
\begin{restatable}{proposition}{houseassignment}
\label{thm:houseassignment}
Assume $\lvert N \rvert \ge 2.$ Every partition $P$ for an individually rational and efficient housing assignment choice rule $\phi$ produces contextual privacy violations. 
\end{restatable}
\begin{proof}
The proof shows that there is collective but not individual pivotality. Consider two agents $i$ and $j$ and two possible preference profiles for each agent. For agent $i$, consider a type $\theta_i$ which contains $e(i) \succ_i e(j),$ and a type $\theta_i'$ which contains $e(j) \succ_i e(i).$ For agent $j$, consider $\theta_j$ which contains $e(j) \succ_j e(i),$ and a type $\theta_j'$ which contains $e(i) \succ_j e(j).$ Hold fixed all other types $\btheta_{-ij},$ to be such that they prefer their own endowment to all other objects, i.e. $\btheta_{-ij} = (e(k) \succeq_k c$ for all $c \in C \setminus \{e(k)\})_{k \in N \setminus\{i,j\}}$.

When the type profile is $(\theta_i, \theta_j, \btheta_{-ij}),$ the agents both prefer their own endowment to the other's. When the profile is $(\theta_i', \theta_j, \btheta_{-ij}),$ or $(\theta_i, \theta_j',\btheta_{-ij}),$ they both prefer $i$'s endowment and $j$'s endowment, respectively. When $(\theta_i', \theta_j', \btheta_{-ij})$, they each prefer the other's endowment to their own. Let $x$ be the outcome in which both agents retain their endowment, i.e. $x=(i, e(i)), (j, e(j))$. Let $y$ be the outcome in which each agent gets each other's endowment $y=(i, e(j)), (j, e(i))$. Then, individual rationality makes the requirements shown on the mid-left in \Cref{fig:square_ir_eff}: $\phi(\theta_i, \theta_j, \btheta_{-ij}) = (\theta_i, \theta_j', \btheta_{-ij})= (\theta_i', \theta_j, \btheta_{-ij})=x$. Meanwhile, efficiency requires $\phi(\theta_i', \theta_j', \btheta_{-ij})=y$ (shown on the mid-right in \Cref{fig:square_ir_eff}). Hence, by \Cref{cor:collectindividual}, every partition for an individually rational and efficient choice rule produces contextual privacy violations.
\end{proof}
\begin{figure}
    \centering
    \includegraphics[scale=.15]{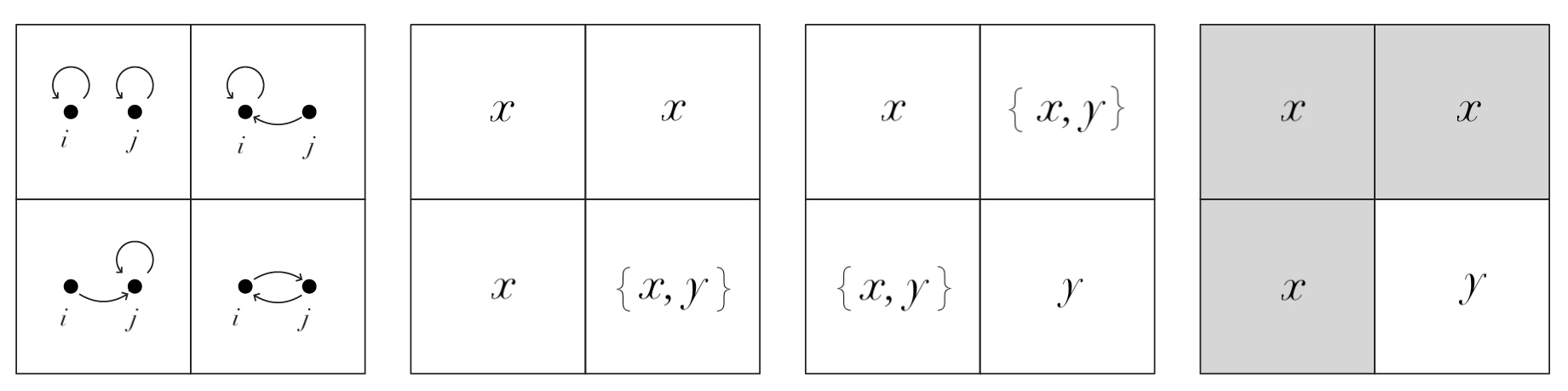}
    \caption{Constructing collective but not individual pivotality for house assignment. Type profiles $\{\theta_i, \theta_i'\} \times \{\theta_j, \theta_j'\}$ used in the proof (left, arrows denote whether agent prefers own or other's endowed object); required outcomes for each type profile under individual rationality (mid-left), under efficiency (mid-right), and under both efficiency and individual rationality (right).}
    \label{fig:square_ir_eff}
\end{figure}

The failure of contextual privacy in the house assignment problem is illuminating. There will be a violation at any type profile $\btheta$ in which there is a pair of agents who each prefer their endowment to the other's. Note that in a setting with many agents and many objects, there are many such type profiles $\btheta$ in $\bTheta$. In such cases, the conjunction of efficiency and individual rationality produce an instance where the agents are collectively but not individually pivotal.

\subsection{Voting}\label{subsec:voting}

A natural case of collective but not individual pivotality arises in voting. If three agents vote on a binary issue, then a single agent cannot certainly change the outcome, but two agents together can. This issue holds more generally for generalized median voting rules. These are defined on an ordered outcome space $(X, \le)$. The type space is the set of single-peaked preferences with respect to $\le$.\footnote{A preference $\preceq$ on $(X, \le)$ is single-peaked if  $x<x'\leq \theta_i \implies x' \succ_i x$ and $x>x'\geq \theta_i \implies x' \succ_i x$.} 

We consider a commonly studied class of voting rules. Namely, we study \emph{generalized median voting rules}. As shown in  \citet{genmedian}, this class is the class of all anonymous, strategy-proof and Pareto-efficient voting rules, where anonymity means that the outcome cannot depend on the identity of any agent. A generalized median voting rule takes as input submitted peaks of agents' preferences $\theta_1, \theta_2, \dots \theta_n$ as well as \emph{phantom ballots} $k_1, k_2, \dots, k_{n-1} \in X \cup \{ -\infty, \infty\}$, where for all $x \in X$, $-\infty < x < \infty$. The output is the median of the submitted votes and phantom votes, i.e.  
\[
\phi_{(k_1,\, k_2,\, \dots,\, k_{n-1})}(\btheta) = \median (\theta_1, \theta_2, \dots, \theta_n, k_1, k_2, \dots, k_{n-1}).
\]
\begin{restatable}{proposition}{semv}\label{prop:se_mv}
Assume $\lvert N \rvert \ge 2$ and $\lvert X \rvert \ge 2$ and the preference domain of single-peaked preferences. Any iterative partition $P$ for a generalized median voting rule that is neither the maximum nor the minimum rule produces contextual privacy violations.
\end{restatable}

In such voting rules, classical examples of collective but not individual pivotality arise from voting thresholds, e.g., a qualified majority. Consider a voting rule with at least two phantom ballots $k_i$ on an alternative $x$ that would win under a type profile $\btheta$. If there are two agents on the right of this alternative, and \emph{both} of them having a type on the left of $x$, it would change the outcome---they are \emph{collectively pivotal}---and yet it is still possible that neither of them alone would change the outcome.

\begin{proof}
    Consider two adjacent types $\theta, \theta' \in \Theta$ and $\btheta_{-ij} \in \bTheta_{-ij}$ such that 
    \[
    \theta = \median (\btheta_{-ij}, k_1, k_2, \dots, k_{n-1})
    \]
    and exactly one type in $\btheta_{-ij}, k_1, k_2, \dots, k_{n-1}$ is $\theta$. The type profiles $(\theta, \theta, \btheta_{-ij})$, $(\theta', \theta, \btheta_{-ij})$ and $(\theta, \theta', \btheta_{-ij})$ will all result in outcome $\theta$. However, $(\theta', \theta', \btheta_{-ij})$ will result in $\theta'$. This hence produces an instance of collective pivotality without individual pivotality and shows that there must be a contextual privacy violation for at least one of the agents $i,j$.
\end{proof}

\section{The Overdescending-Join Protocol}\label{sec:overdesc}
In \Cref{sec:maxcp}, we defined the ascending-join protocol and showed that it is a maximally contextually private protocol for the $k$-item Vickrey auction. In this appendix, we consider a descending protocol $P_{\ODESCJ}$, and use result from \Cref{sec:maxcp} to show that it is maximally contextually private. In the \emph{overdescending-join protocol}, for every type profile $\tilde \bTheta$, the designer asks agents $i=n, n-1, \dots, 1$ whether they can rule out $\max \phi(\tilde \bTheta)$, where, as in the main text, the outcomes are ordered first by price then in the strong set order on the set of winners. We chose the next outcome as the maximal outcome that is not ruled out by an agent.

For the $k$-item Vickrey auction rule $\phik$, $k = 1, 2, \dots, n-1$ we obtain the following two results.
\begin{proposition}\label{prop:overdescjmaxcp}
    The overdescending-join protocol is maximally contextually private for the $k$-item Vickrey auction $\phik$.
\end{proposition}
\begin{proof}
    Let $P$ be a protocol for $\phik$ such that $\Gamma(P, \phik) \subseteq \Gamma(P_{\ODESCJ}, \phik)$. We construct the \emph{reflected} social choice function $\phik^r \colon \bTheta^r \mapsto (W^r, t^r)$, which, for any $\btheta \in \bTheta$ has $i \in W^r(\btheta)$ if and only if $i \notin W(\btheta)$ and $i$ is not the $k$-th highest agent (with ties broken in the order $i= 1, 2, \dots, n$), and $t^r(\btheta) = t(\btheta)$. Define the type space $\Theta' = (\Theta, \le')$ with the reflected order, where $\theta \le' \theta' \iff \theta \ge \theta$. Also define the reflected agent order $\btheta^r = (\theta_n, \theta_{n-1}, \dots, \theta_1)$. Observe that $\phik^r \colon \bTheta' \to 2^{N} \times \bTheta$ is a $(n-k-1)$-item Vickrey auction,
    \begin{equation}
    \phik^r(\btheta^r) = \phi_{V\text{-}(n-k-1)} (\btheta). \label{eq:reflectionscf}
    \end{equation}
    We also define the reflected protocols $P^r$ and $P_{\ODESCJ}^r$, which, whenever $P$ resp. $P_{\ODESCJ}$ queries agent $i$, it queries agent $n - i + 1$, and agent action $\tilde \Theta \subseteq \Theta$ is translated to agent action $\{ \theta^r | \theta \in \tilde \Theta\}$. Observe that 
    \begin{equation}
    P_{\ODESCJ}^r = P_{\ASCJ} \text{ and } P_{\ASCJ}^r = P_{\ODESCJ} \label{eq:protocolreflection}
    \end{equation}
     and that reflection preserves the contextual privacy order. i.e. for any two protocols $P_1$, $P_2$
    \begin{equation}
    \Gamma(P_1, \phi) \subseteq \Gamma(P_2, \phi) \implies \Gamma(P^r_1, \phi^r) \subseteq \Gamma(P^r_2, \phi^r). \label{eq:reflectionandcp}
    \end{equation}
    Hence, as we assume $\Gamma(P, \phik) \subseteq \Gamma(P_{\ODESCJ}, \phik)$, we have by \eqref{eq:reflectionandcp} that 
    \[
    \Gamma(P^r, \phik^r) = \Gamma(P^r, \phi_{V\text{-}(n-k-1)}) \subseteq \Gamma(P_{\ASCJ}, \phi_{V\text{-}(n-k-1)}) =\Gamma(P_{\ODESCJ}^r, \phik^r).   
    \]
    Here, the last equality is by \eqref{eq:reflectionscf}. By \Cref{lemma:maxcpkpa}, it must be that 
    \[
    \Gamma(P^r, \phi_{V\text{-}(n-k-1)}) = \Gamma(P_{\ASCJ}, \phi_{V\text{-}(n-k-1)}),
    \]
hence $\Gamma(P, \phik^r) = \Gamma(P_{\ODESCJ}^r, \phik^r)$, concluding the proof.
\end{proof}
We can express which agents' privacy is protected with the overdescending protocol. Recall that $d(\btheta)$ is the index of the agent determining the price. Define
\begin{align*}
\L &= \{ (\btheta, i) \mid \text{$i$ is neither a winner nor determines the price at $\btheta$}\} \\
\LQW &= \W (\btheta) \cap \{ i \mid i < d(\btheta) \}.
\end{align*}
\begin{proposition}\label{prop:overdescjdelay}
    Let $\btheta \in \bTheta$. The overdescending-join protocol protects losing and late-queried winning bidders. That is, for all $\btheta \in \bTheta$,
    \[
       i \in \L(\btheta) \cup \LQW(\btheta) \Rightarrow (\btheta, i) \notin  \Gamma (\phik, P_{\ODESCJ}).
    \]
\end{proposition}
A naive overdescending protocol which would ask all agents from the top until the $(k+1)$st highest bid is found can protect losers, but does not protect late-queried winners. Delay in querying protects these winners in addition to losers. For simplicity, we provide a direct proof of this result, not using the reflection approach in the proof of \Cref{prop:overdescjmaxcp}.
\begin{proof}
The protocol stops when the $(k+1)$st highest bid is found, hence only revealing the fact that losers have a type at most $\btheta_{[k+1]}$. As the queries are going from $i=n$ down to $1$, the first time that late-queried winners are asked, this means that the type $\btheta_{[k+1]}$ has been found, and them saying that their type is at least $\btheta_{[k+1]}$. This shows that
\[
   (\btheta, i) \notin  \Gamma (\phik, P_{\ODESCJ}) \impliedby i \in \L(\btheta) \cup \LQW(\btheta).
\]
\end{proof}

\section{Relative Informativeness}\label{sec:informativeness}

Our analysis has emphasized contextual privacy, a notion explicitly linked to the social choice rule $\phi$ and sensitive to the privacy concerns of individual agents. However, privacy criteria that are not based on the choice rule can offer insights about informative revelation that are largely complementary to those revealed through a contextual privacy analysis. In particular, as noted in \Cref{sec:whycp}, the notion of \textit{relative informativeness} \citep{mackenzie2020revelation,segal2007communication} can allow a designer to further narrow down their selection of a mechanism within a set of mechanisms that are equivalent in terms of contextual privacy. 

\begin{definition}
    An iterative partition $P$ is \emph{less informative} than an iterative partition $P'$, denoted $P \preceq_{\RI} P'$, if any type profiles $\btheta$, $\btheta'$ that $P$ distinguishes are also distinguished by $P'$. 
\end{definition}
 We note that the relative informativeness partial order is strictly coarser than contextual privacy. 
\begin{proposition}
\label{prop:relativeinformativeness}
    Let $P, P'$ be iterative partitions for $\phi$. Then 
 \begin{itemize}
     \item[(i)] for all $P, P'$ for $\phi$,
      \[
    P \preceq_{\RI} P' \implies \Gamma(P, \phi) \subseteq \Gamma(P',\phi).
    \]
    \item[(ii)] there exist $P, P'$ for $\phi$ such that $P \not\preceq_{\RI} P', P' \not\preceq_{\RI} P$ but $\Gamma(P,\phi) \subset \Gamma(P',\phi).$
 \end{itemize}
\end{proposition}
\begin{proof}
We begin with (i) and prove the contrapositive. Assume that $P \preceq_{\RI} P'$ but $\Gamma(P, \phi) \not\subseteq \Gamma(P',\phi)$. Then, there is $(\btheta, i) \in \Gamma(P, \phi) \setminus \Gamma(P', \phi)$. That means that there is a $\btheta' = (\theta_i', \btheta_{-i})$ that are distinguished by $P$ but not $P'$ (these satisfy $\phi(\btheta) = \phi(\btheta')$, but this is insubstantial for this claim). This means that $P \not\preceq_{\RI} P'$.

For claim (ii), let $n=2$, $\Theta_1=\{0,1\}$, $\Theta_2=\{0,1,2\}$ and let the choice rule be the simple indicator function $\phi(\theta_1,\theta_2)=\1_{\theta_1=1}$, i.e. the outcome is 1 if and only if $\theta_1=1$, otherwise it is zero. 

We consider two iterative partitions. $P$ first asks agent 1 is \enquote{Is $\theta_1=1$?} On an affirmative answer, ask agent $2$ \enquote{Is $\theta_2=0$?}, otherwise stop. This leads to the final partition
\[
P=
\Bigl\{
      \underbrace{\{(0,0),(0,1),(0,2)\}}_{\text{outcome }0},\;
      \underbrace{\{(1,0)\}}_{\text{outcome }1},\;
      \underbrace{\{(1,1),(1,2)\}}_{\text{outcome }1}
\Bigr\}.
\]
A second partition $P'$ again asks agent $1$ \enquote{Is $\theta_1=1$?} On an affirmative answer, ask agent $2$ \enquote{Is $\theta_2=2$?} On a negative answer, ask agent $2$ \enquote{Is $\theta_2=0$?} The final partition for $P'$ is
\[
P'=
\Bigl\{\underbrace{\{(0,0)\}}_{\text{outcome 0}},
      \underbrace{\{(0,1),(0,2)\}}_{\text{outcome }0},\;
      \underbrace{\{(1,0), (1,1)\}}_{\text{outcome }1},\;
      \underbrace{\{(1,2)\}}_{\text{outcome }1}
\Bigr\}.
\]

$P$ distinguishes $(1,0)$ from $(1,1)$ whereas $P'$ does not; conversely, $P'$ distinguishes $(0,0)$ from $(0,1)$ while $P$ does not.  
Thus $P\not\preceq_{\RI}P'$ and $P'\not\preceq_{\RI}P$.

Because $\phi$ ignores agent 2’s type, \emph{any} difference in $\theta_2$
is economically irrelevant.  In $P$ the designer learns $\theta_2$
\emph{only} when $\theta_1=1$, giving
\[
  \Gamma(P,\phi)=\{((1,0),2),\;((1,1),2),\;((1,2),2)\}.
\]
In $P'$ she obtains information about $\theta_2$ for both $\theta_1 = 1$ and $\theta_1 = 0$, hence
\[
  \Gamma(P',\phi)=\Gamma(P,\phi)\cup
  \{((0,0),2),\;((0,1),2),\;((0,2),2)\}.
\]
Hence $\Gamma(P,\phi)\subset\Gamma(P',\phi)$ and therefore
$P\preceq_{CP}P'$. We have found two protocols that are \emph{incomparable} in the relative-informativeness but are strictly ranked by contextual
privacy, completing the proof of part (ii). 
\end{proof}
Relative informativeness can be used to refine and select elements that are contextual privacy equivalent.\footnote{We thank an anonymous reviewer for offering a version of this example.}
\begin{example}{Equivalence under $\preceq_{\CP}$ but comparability under $\preceq_{\RI}$.}
Consider $n = 1$, $\Theta = \{1, 2, \dots, k\}$ and $\phi(\theta) = \1_{\theta = k}$. In this case, there is a unique maximally contextually private protocol $P_{\theta=k}$ which asks the agent \enquote{is your value equal to $k$?} Any other two protocols $P, P'$ for $\phi$ that are not $P_{\theta=k}$ are contextual privacy equivalent to each other. (Formally, all protocols except for the one with the final partition $\{\{k\}, \{1, 2, \dots, k-1\}\}$ are contextual privacy equivalent.) 
    
However, it is clear that some protocols yield less information overall, and thus, for a privacy-conscious designer, they may be more desirable. For example, consider the protocol $P_{\theta=1, \theta=k}$ which asks the agent \enquote{is your value equal to $1$?} and then \enquote{is your value equal to $k$?} This yields a partition $\{\{1\},\{k\}, \{2, \dots, k-1\}\}$. And then consider the protocol $P_{\reveal}$ which asks the agent to reveal her type exactly. This yields a partition $\{\{1\},\{2\},\{3\},\dots, \{k\}\}$. $P_{\theta=1, \theta=k}$ and $P_{\reveal}$ are contextual privacy equivalent. However, there is clearly a sense in which $P_{\reveal}$ reveals more information---indeed, this is captured by relative informativeness, $P_{\theta=1, \theta=k}\prec_{\RI} P_{\reveal}.$
\end{example}
To summarize the discussion thus far, relative informativeness and contextual privacy are complementary. Among a set of protocols that are \emph{equivalent} in contextual privacy order, relative informativeness \emph{can sometimes make further comparisons}. Among a set of protocols that are \emph{incomparable} in the relative informativeness order, contextual privacy \emph{can sometimes compare them}.

This allows us to evaluate the relative informativeness of protocols considered in this paper. Consider first the protocol for the first-price auction rule that asks agents in decreasing order the value of their bid function $b_i(\theta)$ whether their type is above some threshold type $\tilde\theta$. Because of lexicographic tiebreaking, ask in the order $i=1, 2, \dots, n$. We call this protocol $P_{\DESC}$.
\begin{proposition}\label{prop:informativeness_desc}
    $P_{\DESC}$ is minimally relatively informative among protocols computing $\phi_{\FPA}$.
\end{proposition}
\begin{proof}
    The final partition of $P_{\DESC}$ consists of sets with the same values: 
    
    \[\left(\min \argmax_{i = 1, 2 \dots, n} b_i(\theta_i), \max_{i = 1, 2, \dots, n} b_i(\theta_i)\right).\] 
    Hence, in particular, it distinguishes only type profiles that lead to different outcomes. Any coarser partition of $\bTheta$ (not only final partitions of protocols) can therefore not compute $\phi_{\FPA}$.
\end{proof}
The result for the first-price auction rule shows a strong kind of optimality from a privacy perspective. The protocol only distinguishes type profiles that have different outcomes---it is fully contextually private and also a minimal element in the relative informativeness order. 

We can show a similar result for the ascending-join protocol. To show that the ascending-join protocol $P_{\ASCJ}$ is minimally relatively informative, we need two definitions. We say that a protocol $P$ is \emph{flat} if there is no $\tilde \bTheta$ such that for all agents $i=1, 2, \dots, n$, types $\theta_i, \theta_i' \in \tilde \Theta_i$ and  $ \btheta_{-i} \in \tilde \bTheta_{-i}$, $(\theta_i, \btheta_{-i})$ and $(\theta_i', \btheta_{-i})$ are distinguished, but not at $\tilde \bTheta$. That is, a protocol is flat if it doesn't delay distinguishing any agent's types until later stages when that same distinction could have already been made at an earlier stage. We say it is \emph{nonredundant} if for all internal nodes $\tilde \bTheta$, $\phi(\tilde \bTheta)$ is non-constant.

\begin{proposition}\label{prop:flatnonred}
    Let $P$ be a flat and nonredundant protocol for $\phi$. Then $P$ is minimally relatively informative for $\phi$.
\end{proposition}
\begin{proof}
    For contradiction assume that there was $P' \prec_{\RI} P$. Let $\tilde \bTheta$ be an earliest node (with respect to the precedence order of $P'$ and $P$) such that the children of $\tilde \bTheta$ are not identical in $P$ and $P'$. There are two cases. First, $P'$ might terminate at $\tilde \bTheta$, leading to a contradiction with nonredundance. Second, $P'$ might query differently from $P$. As $P'$ is weakly less relatively informative, it must be the case that all type profiles that are distinguished at $\tilde \bTheta$ in $P'$ must be distinguished for all children of $\tilde \bTheta$ in $P$, but not at $\tilde \bTheta$, yielding a contradiction with flatness. Therefore, such a node cannot exist and $P' = P$.
\end{proof}
\begin{proposition}\label{prop:informativeness_ascj}
    $P_{\ASCJ}$ is a flat and non-redundant protocol for $\phik$. In particular, $P_{\ASCJ}$ is minimally relatively informative among protocols computing $\phik$.
\end{proposition}
\begin{proof}
    For bimonotonic protocol like $P_{\ASCJ}$ it is sufficient to show that there is no node $\tilde \bTheta$ such that the same threshold query $(j, t')$ to agent $j$ at threshold $t$ is asked for all $\btheta_{-i} \in \tilde \bTheta_{-i}$. By the structure of $P_{\ASCJ}$ we can identify each node $\tilde \bTheta$ of $P_{\ASCJ}$ with a threshold query $(i, t)$. There are two cases for queries $(j, t')$. 
    
    The first one is where $(j, t')$ is a threshold query for a higher threshold, $(j, t')$, $j = 1, 2, \dots, n$, $t' > t$. These are not queried if the protocol terminates at running price $t$, which is possible if all but $k$ agents $i = 1,2, \dots, n$ have type $\theta_i \le t$. Hence, flatness holds for such nodes $\tilde \bTheta$ and future queries $(j, t')$. A second case is a query to another agent at the same running price, $(t, j)$, $j \in N \setminus \{0, i\}$. We distinguish two cases. If $j < i$, then by \eqref{eq:k1hremaining} if all agents $j' \ge i$ drop out at $t$, $P_{\ASCJ}$ does not query $(t, j)$, but terminates. If $j > i$, then if all agents $j' \le i$ stay in at cutoff $t$, then $(j, t)$ is not queried, but the price is raised \enquote{before it's $j$'s turn.}
\end{proof}
Using symmetry like in the proof of \Cref{prop:overdescjmaxcp} it can be shown that the overdescending-join protocol presented in \Cref{sec:overdesc} is minimally relatively informative.
\end{document}